\documentclass[11pt, letterpaper]{article}
\usepackage{setspace}
\usepackage{amsfonts}
\usepackage{amsmath}
\usepackage{amsthm}
\usepackage{amssymb}
\usepackage{mathrsfs}
\usepackage{mathtools}
\usepackage[hyphens]{url}
\usepackage{hyperref}
\usepackage[x11names]{xcolor}
\usepackage{graphicx}
\usepackage{xspace}
\usepackage{listings}
\usepackage{units}
\usepackage{epigraph}
\usepackage{soul}
\usepackage[capitalize,noabbrev,nameinlink]{cleveref}
\usepackage[margin=1in]{geometry}
\usepackage{microtype}

\newcommand{\mult}[1]{\mathcal M_{#1}}
\newcommand{\hist}[2]{\xi_{#1}({#2})}
\newcommand{\View}[1]{\mathcal V_{#1}}
\newcommand{\mypar}[1]{\noindent\textbf{#1}}
\newcommand{\myskip}{\medskip}

\newtheorem{theorem}{Theorem}[section]
\newtheorem{lemma}[theorem]{Lemma}
\newtheorem{corollary}[theorem]{Corollary}
\newtheorem{observation}[theorem]{Observation}

\lstdefinestyle{mystyle}{
    backgroundcolor=\color{Azure1},   
    basicstyle=\linespread{1.005}\ttfamily\small,
    breakatwhitespace=false,         
    breaklines=true,                 
    captionpos=b,
    frame=single,            
    keepspaces=true,                 
    numbers=left,
    numberstyle=\small,                
	numbersep=5pt,                  
    showspaces=false,                
    showstringspaces=false,
    showtabs=false,                  
    tabsize=2
}

\title{Computing in Anonymous Dynamic Networks Is Linear}
\author{
Giuseppe A. Di Luna\thanks{
DIAG, Sapienza University of Rome, \texttt{g.a.diluna@gmail.com}}
\and
Giovanni Viglietta\thanks{
University of Aizu, \texttt{viglietta@gmail.com}}
}
\date{}
\begin{document}

\maketitle

\setlength\epigraphwidth{.3\textwidth}
\epigraph{There is no permanent place in the world for ugly mathematics.}{\textit{G.\,H.~Hardy, 1940}}

\begin{abstract}
We give the first linear-time counting algorithm for processes in anonymous 1-interval-connected dynamic networks with a leader. As a byproduct, we are able to compute in $3n$ rounds \emph{every} function that is deterministically computable in such networks. If explicit termination is not required, the running time improves to $2n$ rounds, which we show to be optimal up to a small additive constant (this is also the first non-trivial lower bound for counting). As our main tool of investigation, we introduce a combinatorial structure called \emph{history tree}, which is of independent interest. This makes our paper completely self-contained, our proofs elegant and transparent, and our algorithms straightforward to implement.

In recent years, considerable effort has been devoted to the design and analysis of counting algorithms for anonymous 1-interval-connected networks with a leader. A series of increasingly sophisticated works, mostly based on classical mass-distribution techniques, have recently led to a celebrated counting algorithm in $O({n^{4+ \epsilon}} \log^{3} (n))$ rounds (for $\epsilon>0$), which was the state of the art prior to this paper. Our contribution not only opens a promising line of research on applications of history trees, but also demonstrates that computation in anonymous dynamic networks is practically feasible and far less demanding than previously conjectured.
\end{abstract}

\clearpage

\tableofcontents

\clearpage

\section{Introduction}\label{xs:1}

\mypar{Background.} The study of theoretical and practical aspects of highly dynamic distributed systems has received a great deal of attention~\cite{CFQS12,KO,MS18}. These models involve a constantly changing network of computational devices called \emph{processes} (sometimes referred to as ``processors'' or ``agents''). This dynamism is typical of modern real-world systems and is the result of technological innovations, such as the spread of mobile devices, software-defined networks, wirelesses sensor networks, wearable devices, smartphones, etc.

There are several models of dynamism~\cite{CFQS12}; a popular choice is the \emph{1-interval-connected} network model~\cite{KLO10, DW05}. Here, a fixed set of $n$ processes communicate through links forming a time-varying graph, i.e., a graph whose edge set changes at discrete time units called \emph{rounds} (thus, the system is synchronous); such a graph changes unpredictably, but is assumed to be connected at all times.

A large number of research papers have considered dynamic systems where each process has a distinct identity (\emph{unique IDs})~\cite{KLO10}. In this setting, there are efficient algorithms for consensus~\cite{KOM11}, broadcast~\cite{CFMS15}, counting~\cite{KLO10, DW05}, and many other problems~\cite{KLO11,MS18}.

The study of dynamic 1-interval-connected networks with unique IDs began with a seminal paper by Kuhn et al.~\cite{KLO10}, who showed that knowing the size of the system (i.e., the number $n$ of processes) is useful for non-trivial computations. For example, to find out if there is at least one process with input $x$ in the system, every process can start broadcasting its input and wait $n-1$ rounds to see if it receives $x$. Note that this technique is ineffective if nothing is known about $n$.

\smallskip
\mypar{Counting problem.}  There are real-world scenarios in which the individual processes may be unaware of the size of the system; an example are large-scale ad-hoc sensor networks~\cite{DW05}. In such scenarios, the problem of determining $n$ is called \emph{Counting problem}. In~\cite{KLO10} it is shown how to solve the Counting problem in at most $n+1$ rounds in a 1-interval-connected network with unique IDs.

\smallskip
\mypar{Anonymous systems.} The known landscape changes if we consider networks without IDs, which are called \emph{anonymous systems}. In this model, all processes have identical initial states, and may only differ by their inputs. In the last thirty years, a large body of works~\cite{BV01,CDS06,CGM08,JMM12,FPP00,SUW15,YK88} have investigated the computational power of anonymous static networks, giving characterizations of what can be computed in various settings~\cite{BV01,BV02}. Studying anonymous systems is not only important from a theoretical perspective, but also for their practical relevance. In a highly dynamic system, IDs may not be guaranteed to be unique due to operational limitations~\cite{DW05}, or may compromise user privacy. Indeed, users may not be willing to be tracked or to disclose information about their behavior; examples are COVID-19 tracking apps~\cite{SM20}, where a threat to privacy was felt by a large share of the public even if these apps were assigning a rotating random ID to each user. In fact, an adversary can easily track the continuous broadcast of a fixed random ID tracing the movements of a person~\cite{LAOOO20}. Anonymity is also found in insect colonies and other biological systems~\cite{GMTN15}.

\smallskip
\mypar{Unique leader.} In order to deterministically solve non-trivial problems in anonymous systems, it is necessary to have some form of initial ``asymmetry''~\cite{A80,BV02,MCS13,YK88}. The most common assumption is the existence of a \emph{leader}, i.e., a single special process that starts in a different and unique initial state. The presence of a leader is a realistic assumption: examples include a base station in a mobile network, a gateway in a sensor network, etc. For these reasons, the computational power of anonymous systems enriched with a leader has been extensively studied in the classical model of static networks~\cite{FPP00,Sa99,YK96}, as well as in population protocols~\cite{AAE08,ABBS16,BBCD15,BBK11,DFIISV19}.

\smallskip
\mypar{State of the art.} A long series of papers have studied the Counting problem in anonymous 1-interval-connected networks with a leader~\cite{CMM16,DB15,DB16,DBBC14a,DBBC14b,DBCB13,KM18,KM19,KM22,MCS13}. These works have shown better and better upper bounds, leading to the recent paper~\cite{KM22}, which solves the Counting problem in $O({n^{4+ \epsilon}} \log^{3} (n))$ rounds (for $\epsilon>0$). Almost all of these works share the same basic approach of implementing a mass-distribution mechanism similar to the local averaging used to solve the average consensus problem~\cite{C11,OT09,T84}. (In \cref{as:4} we give a comprehensive survey.) We point out that the mass-distribution approach requires processes to exchange numbers whose representation size grows at least linearly with the number of rounds. Thus, all previous works on counting in anonymous 1-interval-connected networks need messages of size at least $\Omega(n^4)$.

In spite of the technical sophistication of this line of research, there is still a striking gap in terms of running time---a multiplicative factor of $O(n^3 \log ^{3}(n))$---between the best algorithm for anonymous networks and the best algorithm for networks with unique IDs. The same gap exists with respect to \emph{static} anonymous networks, where the Counting problem is known to be solvable in $2n$ rounds~\cite{MCS13}. Given the current state of the art, solving non-trivial problems in large-scale dynamic networks is still impractical.

\subsection{Our Contributions}
\mypar{Main results.} In this paper, we close the aforementioned gaps by showing that counting in 1-interval-connected anonymous dynamic networks with a leader is linear: we give a deterministic algorithm for the Counting problem that terminates in $3n-2$ rounds (\cref{xth:term}), as well as a non-terminating algorithm that stabilizes on the correct count in $2n-2$ rounds (\cref{xth:stab}).

We also prove a lower bound of roughly $2n$ rounds, both for terminating and stabilizing counting algorithms (\cref{xth:lower}); this is the first non-trivial lower bound for anonymous networks (better than $n-1$), and it shows that our stabilizing algorithm is optimal up to a small additive constant.

In addition, our algorithms support network topologies with multiple parallel links and self-loops (i.e., \emph{multigraphs}, as opposed to the simple graphs used in traditional models).

\myskip
\mypar{Significance.} Our algorithms actually solve a generalized version of the Counting problem: when processes are assigned inputs, they are able to count how many processes have each input. Solving such a \emph{Generalized Counting problem} allows us to solve a much larger class of problems, called \emph{multi-aggregation problems}, in the same number of rounds (\cref{xth:compl}). On the other hand, these are the only problems that can be solved deterministically in anonymous networks (\cref{xth:imp}).

Thus, we come to an interesting conclusion: In 1-interval-connected anonymous dynamic networks with a leader, \ul{any problem that can be solved deterministically has a solution in at most $3n-2$ rounds}. Our lower bounds show that there is an overhead to be paid for counting in anonymous dynamic networks compared to networks with IDs, but the overhead is only linear. This is much less than what was previously conjectured~\cite{KM20,KM22}; in fact, our results make computations in anonymous and dynamic large-scale networks possible and efficient in practice.

We remark that the local computation time and the amount of processes' internal memory required by our algorithms is only polynomial in the size of the network. Also, like in previous works, processes need to send messages of polynomial size.

\myskip
\mypar{Technique.} Our algorithms and lower bounds are based on a novel combinatorial structure called \emph{history tree}, which completely represents an anonymous dynamic network and naturally models the idea that processes can be distinguished if and only if they have different ``histories'' (\cref{xs:3}). Thanks to the simplicity of our technique, this paper is entirely self-contained, our proofs are transparent and easy to understand, and our algorithms are elegant and straightforward to implement.\footnote{An implementation can be found here: \url{https://github.com/viglietta/Dynamic-Networks}. The repository includes a dynamic network simulator that can be used to run tests and visualize the history trees of custom networks.} 

We argue that history trees are of independent interest, as they could help in the study of different network models, as well as networks with specific restricted dynamics.

\section{Definitions and Preliminaries}\label{xs:2}
Here we give some informal definitions; the interested reader may find rigorous ones in \cref{as:1}.

\myskip
\mypar{Dynamic network.} Our model of computation involves a system of $n$ \emph{anonymous} processes in a \emph{1-interval-connected} dynamic network whose topology changes unpredictably at discrete time units called \emph{rounds}. That is, at round~$t\geq 1$, the network's topology is modeled by a connected undirected multigraph $G_t$ whose edges are called \emph{links}. Processes can send \emph{messages} through links and can update their \emph{internal states} based on the messages they receive.

\myskip
\mypar{Round structure.} Each round is subdivided into two phases: in the \emph{message-passing phase}, each process broadcasts a message containing (an encoding of) its internal state through all the links incident to it. After all messages have been exchanged, the \emph{local-computation phase} occurs, where each process updates its own internal state based on a function $\mathcal A$ of its current state and the multiset of messages it has just received. Processes are \emph{anonymous}, implying that the function $\mathcal A$ is the same for all of them. We stress that all local computations are deterministic.

\myskip
\mypar{Input and output.} At round~$0$, each process is assigned an \emph{input}, which is deterministically converted into the process' initial internal state. Furthermore, at every round, each process produces an \emph{output}, which is a function of its internal state. Some states are \emph{terminal}; once a process enters a terminal state, it can no longer change it.

\myskip
\mypar{Stabilization and termination.} A system is said to \emph{stabilize} if the outputs of all its processes remain constant from a certain round onward; note that a process' internal state may still change even when its output is constant. If, in addition, all processes reach a terminal state, the system is said to \emph{terminate}.

\myskip
\mypar{Problems.} A \emph{problem} defines a relationship between processes' inputs and outputs: an algorithm $\mathcal A$ \emph{solves} a given problem if, whenever the processes are assigned a certain multiset of $n$ inputs, and all processes execute $\mathcal A$ in their local computations, the system eventually stabilizes on the correct multiset of $n$ outputs. Note that the same algorithm $\mathcal A$ must work for all $n\geq 1$ (i.e., the system is unaware of its own size) and regardless of the network's topology (as long as it is 1-interval-connected). A stronger notion of solvability requires that the system not only stabilizes but actually terminates on the correct multiset of outputs.

\myskip
\mypar{Unique leader.} A commonly made assumption is the presence of a \emph{unique leader} in the system. This can be modeled by assuming that each process' input includes a \emph{leader flag}, and an input assignment is valid if and only if there is exactly one process whose input has the leader flag set.

\myskip
\mypar{Multi-aggregation problems.} A \emph{multi-aggregation problem} is a problem where the output to be computed by each process only depends on the process' own input and the multiset of all processes' inputs. In the special case where all processes have to compute the same output, we have an \emph{aggregation problem}. Notable examples of aggregation problems include computing statistical functions on input numbers, such as sum, average, maximum, median, mode, variance, etc.

As we will see, the multi-aggregation problems are precisely the problems that can be solved in 1-interval-connected anonymous dynamic networks with a unique leader.
Specifically, in \cref{xs:4} we will show that all multi-aggregation problems can be solved in a linear number of rounds. On the other hand, in \cref{xs:5} we will prove that no other problem can be solved at all.

\myskip
\mypar{Counting problem.} An important aggregation problem is the \emph{Generalized Counting problem}, where each process must output the multiset of all processes' inputs. That is, the system has to count how many processes have each input. In the special case where all (non-leader) processes have the same input, this reduces to the \emph{Counting problem}: determining the size of the system, $n$.

\myskip
\mypar{Completeness.} The Generalized Counting problem is \emph{complete} for the class of multi-aggregate problems (no matter if the network is connected or has a unique leader), in the following sense:

\begin{theorem}\label{xth:compl}
If the Generalized Counting problem can be solved (with termination) in $f(n)$ rounds, then every multi-aggregate problem can be solved (with termination) in $f(n)$ rounds, too.
\end{theorem}
\begin{proof}
Once a process with input $x$ has determined the multiset $\mu$ of all processes' inputs, it can immediately compute any function of $(x, \mu)$ within the same local-computation phase.
\end{proof}

\section{History Trees}\label{xs:3}
We introduce \emph{history trees} as a natural tool of investigation for anonymous dynamic networks. An example of a history tree is illustrated in \cref{xfig:hg}, and a formal definition is found in \cref{as:2}.

\myskip
\mypar{Indistinguishable processes.} Since processes are anonymous, they can only be distinguished by their inputs or by the multisets of messages they have received. This leads to an inductive definition of \emph{indistinguishability}: Two processes are indistinguishable at the end of round~$0$ if and only if they have the same input. At the end of round~$t\geq 1$, two processes $p$ and $q$ are indistinguishable if and only if they were indistinguishable at the end of round~$t-1$ and, for every equivalence class $A$ of processes that were indistinguishable at the end of round~$t-1$, both $p$ and $q$ receive an equal number of (identical) messages from processes in $A$ at round~$t$.

\myskip
\mypar{Levels of a history tree.} A \emph{history tree} is a structure associated with a dynamic network. It is an infinite graph whose nodes are subdivided into \emph{levels} $L_{-1}$, $L_0$, $L_1$, $L_2$, \dots, where each node in level $L_t$, with $t\geq 0$, represents an equivalence class of processes that are indistinguishable at the end of round~$t$. The level $L_{-1}$ contains a unique node $r$, representing all processes in the system. Each node in level $L_0$ has a label indicating the (unique) input of the processes it represents.

\myskip
\mypar{Black and red edges.} A history tree has two types of undirected edges; each edge connects nodes in consecutive levels. The \emph{black edges} induce an infinite tree rooted at $r$ and spanning all nodes. A black edge $\{v, v'\}$, with $v\in L_{t}$ and $v'\in L_{t+1}$, indicates that the \emph{child node} $v'$ represents a subset of the processes represented by the \emph{parent node} $v$.

The \emph{red multiedges} represent messages. A red edge $\{v, v'\}$ with multiplicity $m$, with $v\in L_{t}$ and $v'\in L_{t+1}$, indicates that, at round~$t+1$, each process represented by $v'$ receives a total of $m$ (identical) messages from processes represented by $v$.

\myskip
\mypar{Anonymity of a node.} The \emph{anonymity} $a(v)$ of a node $v$ of a history tree is defined as the number of processes represented by $v$. Since the nodes in a same level represent a partition of all the processes, the sum of their anonymities must be $n$. Moreover, by the definition of black edges, the anonymity of a node is equal to the sum of the anonymities of its children.

Observe that the Generalized Counting problem can be rephrased as the problem of determining the anonymities of all the nodes in $L_0$.

\begin{figure}
\centering
\includegraphics[width=\linewidth]{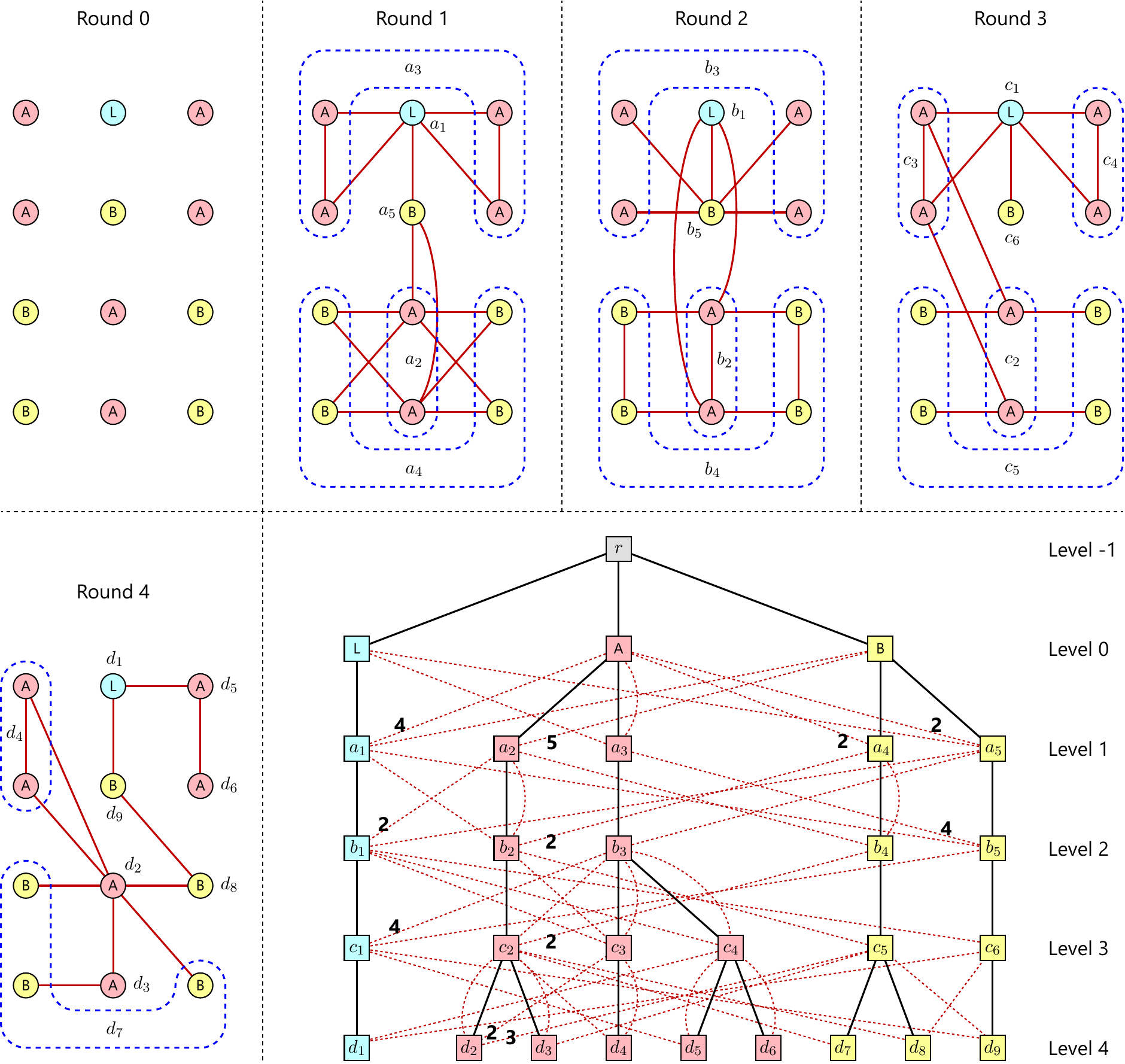}
\caption{The first rounds of a dynamic network with $n=12$ processes and the corresponding levels of the history tree. The letters L, A, B indicate processes' inputs (the process with input L is the leader). At each round, the network's links are represented by solid red edges. Sets of indistinguishable processes are indicated by dashed blue lines and unique labels. The same labels are also reported in the history tree: each node corresponds to a set of indistinguishable processes. It should be noted that labels other than L, A, B are not part of the history tree itself, and have been added for the reader's convenience. The dashed red edges in the history tree represent messages received by processes through the links; the numbers indicate their multiplicities (when greater than $1$). For example, the edge $\{c_2, b_4\}$ has multiplicity~$2$ because, at round~3, each of the two processes in the class~$c_2$ receives two identical messages from the (indistinguishable) processes that were in the class~$b_4$ in the previous round. On the other hand, the node labeled $c_5$ represents four processes (i.e., its \emph{anonymity} is $4$) at round~3; among them, only one receives a message from $c_6$ in the next round. Thus, this process is disambiguated, which causes the node $c_5$ to branch into two nodes: $d_7$, with anonymity~$3$, and $d_8$, with anonymity~$1$. The \emph{view} of the node labeled $b_3$ is the subgraph of the history tree induced by the nodes with labels in the set $\{b_3, a_3, a_5, \mbox{L}, \mbox{A}, \mbox{B}, r\}$.}
\label{xfig:hg}
\end{figure}

\mypar{History of a process.} A \emph{monotonic path} in the history tree is a sequence of nodes in distinct levels, such that any two consecutive nodes are connected by a black or a red edge. We define the \emph{view} of a node $v$ in the history tree as the finite subgraph induced by all the nodes spanned by monotonic paths with endpoints $v$ and $r$. The \emph{history} of a process $p$ at round~$t$ is the view of the (unique) node in $L_t$ that represents a set of processes containing $p$.

\myskip
\mypar{Fundamental theorem.} Intuitively, the history of a process at round~$t$ contains all the information that the process can use at that round for its local computations. This intuition is made precise by the following fundamental theorem (a rigorous proof is found in \cref{as:2.3}):
\begin{theorem}\label{xth:view}
The state of a process $p$ at the end of round~$t$ is determined by a function $\mathcal F_\mathcal A$ of the history of $p$ at round~$t$. The mapping $\mathcal F_\mathcal A$ depends entirely on the local algorithm $\mathcal A$ and is independent of $p$.\qed
\end{theorem}

\mypar{Locally constructing the history.} The significance of \cref{xth:view} is that it allows us to shift our focus from dynamic networks to history trees. As shown in \cref{as:2.4}, there is a local algorithm $\mathcal A^\ast$ that allows processes to construct and update their history at every round. Thanks to \cref{xth:view}, processes are guaranteed not to lose any information if they simply execute $\mathcal A^\ast$, regardless of their goal, and then compute their outputs as a function of their history. Thus, in the following, we will assume without loss of generality that a process' internal state at every round, as well as all the messages it broadcasts, always coincide with its history at that round.

\section{Linear-Time Computation}\label{xs:4}
As discussed at the end of \cref{xs:3}, the following algorithms assume that all processes have their current history as their internal state and broadcast their history through all available links at every round. We will show how to solve the Generalized Counting problem in such a setting.

\subsection{Stabilizing Algorithm}\label{xs:4.1}
In a history tree (or in a view), a node $v$ is said to be \emph{non-branching} if it has exactly one child, which we denote as $c(v)$. Recall that the parent-child relation is determined by black edges only.

A pair of non-branching nodes $(v_1, v_2)$ in a history tree is said to be \emph{exposed} with multiplicity $(m_1, m_2)$ if the red edge $\{c(v_1), v_2\}$ is present with multiplicity $m_1\geq 1$, while the red edge $\{c(v_2), v_1\}$ has multiplicity $m_2\geq 1$ (see~\cref{xfig:mix}, left). Note that $v_1$ and $v_2$ must be on the same level.

Thanks to the following lemma, if we know the anonymity of a node in an exposed pair, we can determine the anonymity of the other node. (Recall that we denote the anonymity of $v$ by $a(v)$.)

\begin{lemma}\label{xl:guess1}
If $(v_1, v_2)$ is an exposed pair with multiplicity $(m_1, m_2)$, then $a(v_1)\cdot m_1=a(v_2)\cdot m_2$.
\end{lemma}
\begin{proof}
Let $v_1, v_2\in L_t$, and let $P_1$ and $P_2$ be the sets of processes represented by $v_1$ and $v_2$, respectively. Since $v_1$ is non-branching, we have $a(c(v_1))=a(v_1)$, and therefore $c(v_1)$ represents $P_1$, as well. Hence, the number of links between $P_1$ and $P_2$ in $G_{t+1}$ (counted with their multiplicities) is $a(c(v_1))\cdot m_1 = a(v_1)\cdot m_1$. By a symmetric argument, this number is equal to $a(v_2)\cdot m_2$.
\end{proof}

It is well known that, in a 1-interval-connected network, any piece of information can reach all processes in $n-1$ rounds~\cite{KLO10}. We can rephrase this observation in the language of history trees.

\begin{lemma}\label{xl:propamain} Let $P$ be a set of processes in a 1-interval-connected dynamic network of size $n$, such that $1\leq |P|\leq n-1$, and let $t\geq 0$. Then, at every round~$t'\geq t+|P|$, in the history of every process there is a node at level $L_t$ representing at least one process not in $P$.
\end{lemma}
\begin{proof}
Let $Q$ be the complement of $P$ (note that $Q$ is not empty), and let $Q_{t+i}$ be the set of processes represented by the nodes in $L_{t+i}$ whose view contains a node in $L_t$ representing at least one process in $Q$. We will prove by induction that $|Q_{t+i}|\geq |Q|+i$  for all $0\leq i\leq |P|$. The base case holds because $Q_t=Q$. The induction step is implied by $Q_{t+i}\subsetneq Q_{t+i+1}$, which holds for all $0\leq i < |P|$ as long as $|Q_{t+i}|<n$. Indeed, because $G_{t+i+1}$ is connected, it must contain a link between a process $p\in Q_{t+i}$ and a process $q\notin Q_{t+i}$. Thus, the history of $q$ at round~$t+i+1$ contains the history of $p$ at round~$t+i$, and so $Q_{t+i}\subsetneq Q_{t+i}\cup\{q\}\subseteq Q_{t+i+1}$.

Now, plugging $i:=|P|$, we get $|Q_{t+|P|}|=n$. In other words, the history of each node in $L_{t+|P|}$ (and hence in subsequent levels) contains a node at level $L_t$ representing a process in $Q$.
\end{proof}

\begin{corollary}\label{xl:propa}
In the history tree of a 1-interval-connected dynamic network, every node at level $L_t$ is in the view of every node at level $L_{t'}$, for all $t'\geq t+n-1$.
\end{corollary}
\begin{proof}
Let $v\in L_t$, and let $P$ be the set of processes not represented by $v$. If $P$ is empty, then all nodes in $L_{t'}$ are descendants of $v$, and have $v$ in their view. Otherwise, $1\leq|P|\leq n-1$, and \cref{xl:propamain} implies that $v$ is in the view of all nodes in $L_{t'}$.
\end{proof}

\begin{theorem}\label{xth:stab}
There is an algorithm that solves the Generalized Counting problem in 1-interval-connected anonymous dynamic networks with a leader and stabilizes in at most $2n-2$ rounds.
\end{theorem}
\begin{proof}
Let $n>1$, and let $L_T$ be the first level of the history tree whose nodes are all non-branching. Since $|L_{0}|\geq 2$ and $|L_{t-1}|\leq |L_t|\leq n$ for all $t\geq 0$, we have $T\leq n-2$ by the pigeonhole principle.

Let $\ell\in L_T$ be the node corresponding to the leader; we know that $a(\ell)=1$. Since $G_{T+1}$ is connected, the exposed pairs of nodes in $L_T$ must form a connected graph on $L_T$, as well. Hence, thanks to \cref{xl:guess1}, any process that has complete knowledge of $L_T$ and $L_{T+1}$ can use the information that $a(\ell)=1$ to compute the anonymities of all nodes in $L_T$. In fact, by \cref{xl:propa}, every process is able to do so by round~$(T+1)+n-1\leq 2n-2$.

The local algorithm for each process is given in \cref{xl:algorithm1}: Find the first level in your history whose nodes are all non-branching. If such a level exists and contains a node corresponding to the leader, assume this is level $L_T$ and compute the anonymities of all its nodes. Then add together the anonymities of nodes having equal inputs, and give the result as output. Even if some outputs may be incorrect at first, the whole system stabilizes on the correct output by round~$2n-2$.
\end{proof}
If we were to use a strategy such as the one in \cref{xl:algorithm1} to devise a terminating algorithm, we would be bound to fail, as the counterexample in \cref{as:3.1} shows. The problem is that a process has no easy way of knowing whether the first level in its history whose nodes are all non-branching is really $L_T$, and so it may end up terminating too soon with the wrong output. To find a correct termination condition, we will have to considerably develop the theory of history trees.

\begin{figure}
\centering
\includegraphics[scale=0.75]{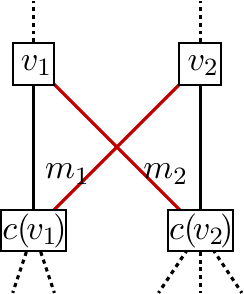}\qquad\quad
\includegraphics[scale=0.75]{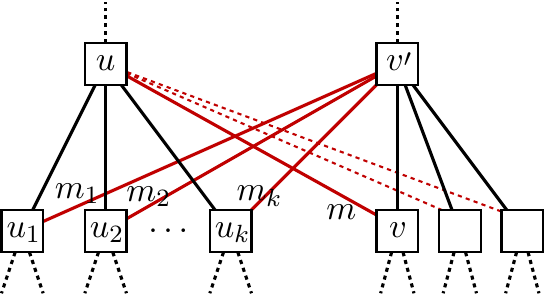}\qquad\quad
\includegraphics[scale=0.75]{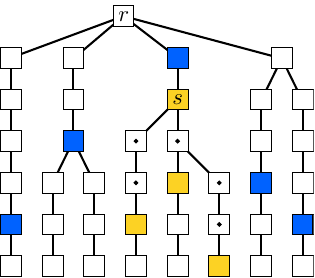}
\caption{Left: an \emph{exposed} pair of nodes. Center: if the anonymities of $u$, $u_1$, $u_2$, \dots, $u_k$ are known, then $v$ is \emph{guessable} by $u$. Right: if the colored nodes are counted, the blue ones form a \emph{counting cut}, and the orange ones define a non-trivial \emph{isle} with root $s$, where the nodes with a dot are internal.}
\label{xfig:mix}
\end{figure}

\lstset{style=mystyle}
\begin{lstlisting}[caption={Algorithm for the Generalized Counting problem that stabilizes in $2n-2$ rounds\label{xl:algorithm1}},captionpos=t,float,abovecaptionskip=-\medskipamount,mathescape=true]
# This local algorithm is executed at every round by each process.
# Input: the current history $\mathcal V$ of the process
# Output: a multiset of inputs received by processes at round $0$

For $t := 0$ to the height of $\mathcal V$
   If level $L_t$ does not contain a leader node, return $\emptyset$
   Assign $a(\ell) := 1$, where $\ell\in L_t$ is the leader node
   Assign ${\rm non\_branching}:={\rm true}$
   For each node $v\in L_t$
      If $v$ does not have exactly one child in $\mathcal V$, assign ${\rm non\_branching}:={\rm false}$
      If $v\neq \ell$, assign $a(v) := 0$
   If ${\rm non\_branching}$
      While there is an exposed pair $(v_1,v_2)$ in $L_t$ with $a(v_1)\neq 0$ and $a(v_2)=0$
         Let $(m_1,m_2)$ be the multiplicity of the exposed pair $(v_1,v_2)$
         Assign $a(v_2):= \left\lceil a(v_1)\cdot m_1/m_2\right\rceil$
      For each node $v\in L_0$
         Assign $a(v):=\sum_{v'\in L_t\text{ descendant of }v} a(v')$
      Return the multiset $\{ ({\rm label}(v), a(v)) \mid v\in L_0 \}$
\end{lstlisting}

\lstset{style=mystyle}
\begin{lstlisting}[caption={Algorithm for the Generalized Counting problem that terminates in $3n-2$ rounds\label{xl:algorithm2}},captionpos=t,float,abovecaptionskip=-\medskipamount,mathescape=true]
# This local algorithm is executed at every round by each process.
# Input: the current history $\mathcal V$ of the process
# Output: either a multiset of inputs received by processes at round $0$
#         or "Unknown"

For each leader node $\ell$ in $\mathcal V$
   Assign $a(\ell):=1$ and mark $\ell$ as counted
While $\mathcal V$ has guessable levels
   Let $v$ be a guessable non-counted node of smallest depth in $\mathcal V$
   Assign a guess $g(v)$ to $v$ as in $\text{\cref{xe:guess}}$ and mark $v$ as guessed
   Let $P_v$ be the black path from $v$ to the root $r$ of $\mathcal V$
   If there is a heavy node in $P_v$
      Let $v'$ be the heavy node in $P_v$ of maximum depth
      Assign $a(v'):=g(v')$; mark $v'$ as counted and not guessed
      If $v'$ is the root or a leaf of a non-trivial complete isle $I$
         For each internal node $w$ of $I$
            Assign $a(w)=\sum_{w'\text{ leaf of }I\text{ and descendant of }w} a(w')$
            Mark $w$ as counted and not guessed
Assign $C:=\emptyset$
For each node $v$ in $\mathcal V$ marked as counted
   Let $P_v$ be the black path from $v$ to the root $r$ of $\mathcal V$
   Let $v'$ be the counted node in $P_v$ of minimum depth
   Assign $C:=C\cup \{v'\}$
If $C$ is not a counting cut of $\mathcal V$, return "Unknown"
Let $L_t$ be the level of $\mathcal V$ containing the deepest node of $C$
Let $L_{t'}$ be the deepest level of $\mathcal V$
Let $n'=\sum_{v\in C}a(v)$
If $t'<t+n'$, return "Unknown"
For each node $v\in L_0$
   Assign $a(v):=\sum_{v'\in C\text{ descendant of }v} a(v')$
Return the multiset $\{ ({\rm label}(v), a(v)) \mid v\in L_0 \}$ and enter a terminal state
\end{lstlisting}

\subsection{Terminating Algorithm}\label{xs:4.2}

\mypar{Guessing anonymities.} Inspired by \cref{xl:guess1}, we now describe a more sophisticated way of estimating the anonymity of a node based on known anonymities (see~\cref{xfig:mix}, center). Let $u$ be a node of a history tree, and assume that the anonymities of all its children $u_1$, $u_2$, \dots, $u_k$ are known: such a node $u$ is called a \emph{guesser}. If $v$ is not among the children of $u$ but is at their same level, and the red edge $\{v, u\}$ is present with multiplicity $m\geq 1$, we say that $v$ is \emph{guessable} by $u$. In this case, we can make a \emph{guess} $g(v)$ on the anonymity of $v$:
\begin{equation}\label{xe:guess}
g(v)=\left\lceil \frac{a(u_1)\cdot m_1+a(u_2)\cdot m_2+\dots+a(u_k)\cdot m_k}m\right\rceil,
\end{equation}
where $m_i$ is the multiplicity of the red edge $\{u_i, v'\}$ for all $1\leq i\leq k$, and $v'$ is the parent of $v$ (possibly, $m_i=0$). Although a guess may be inaccurate, it never underestimates the anonymity:

\begin{lemma}\label{xl:guess2}
If $v$ is guessable, then $g(v)\geq a(v)$. Moreover, if $v$ has no siblings, $g(v)=a(v)$.
\end{lemma}
\begin{proof}
Let $u,v'\in L_t$, and let $P_1$ and $P_2$ be the sets of processes represented by $u$ and $v'$, respectively. By counting the links between $P_1$ and $P_2$ in $G_{t+1}$ in two ways, we have $\sum_i a(u_i)\,m_i = \sum_i a(v_i)\,m'_i$, where the two sums range over all children of $u$ and $v'$, respectively (note that $v=v_j$ for some $j$), and $m'_i$ is the multiplicity of the red edge $\{v_i,u\}$ (so, $m=m'_j$). Our lemma easily follows.
\end{proof}

\mypar{Heavy nodes.} Even if a node is guessable, it is not always a good idea to actually assign it a guess. For reasons that will become apparent in \cref{xl:limit}, our algorithm will only assign guesses in a \emph{well-spread} fashion, i.e., in such a way that at most one node per level is assigned a guess.

Suppose now that a node $v$ has been assigned a guess. We define its \emph{weight} $w(v)$ as the number of nodes in the subtree hanging from $v$ that have been assigned a guess (this includes $v$ itself). Recall that subtrees are determined by black edges only. We say that $v$ is \emph{heavy} if $w(v)\geq g(v)$.

\begin{lemma}\label{xl:limit}
In a well-spread assignment of guesses, if $w(v)>a(v)$, then some descendants of $v$ are heavy (the \emph{descendants} of $v$ are the nodes in the subtree hanging from $v$ other than $v$ itself).
\end{lemma}
\begin{proof}
Our proof is by well-founded induction on $w(v)$. Assume for a contradiction that no descendants of $v$ are heavy. Let $v_1$, $v_2$, \dots, $v_k$ be the ``immediate'' descendants of $v$ that have been assigned guesses. That is, for all $1\leq i\leq k$, no internal nodes of the black path with endpoints $v$ and $v_i$ have been assigned guesses (observe that $k\geq 1$ because, by assumption, $w(v)>1$).

By the basic properties of history trees, $a(v)\geq \sum_i a(v_i)$. Also, the induction hypothesis implies that $w(v_i) \leq a(v_i)$ for all $1\leq i\leq k$, or else one of the $v_i$'s would have a heavy descendant. Therefore, $w(v)-1 = \sum_i w(v_i) \leq \sum_i a(v_i)\leq a(v) \leq w(v)-1$. It follows that $w(v_i)=a(v_i)$ and $a(v)=\sum_i a(v_i)$.

Let $v_d$ be the deepest of the $v_i$'s, which is unique, since the assignment of guesses is well spread. Note that $v_d$ has no siblings at all, otherwise we would have $a(v)>\sum_i a(v_i)$. Due to \cref{xl:guess2}, we conclude that $g(v_d)=a(v_d)=w(v_d)$, and so $v_d$ is heavy.
\end{proof}

\mypar{Correct guesses.} We say that a node $v$ has a \emph{correct} guess if $v$ has been assigned a guess and $g(v)=a(v)$. The next lemma gives a criterion to determine if a guess is correct.

\begin{lemma}\label{xl:crit}
In a well-spread assignment of guesses, if a node $v$ is heavy and no descendant of $v$ is heavy, then $v$ has a correct guess.
\end{lemma}
\begin{proof}
Because $v$ is heavy, $g(v)\leq w(v)$. Since $v$ has no heavy descendants, \cref{xl:limit} implies $w(v)\leq a(v)$. Also, by \cref{xl:guess2}, $a(v)\leq g(v)$. We conclude that $g(v)\leq w(v)\leq a(v)\leq g(v)$, and therefore $g(v)=a(v)$.
\end{proof}
When the criterion in \cref{xl:crit} applies to a node $v$, we say that $v$ has been \emph{counted}. So, counted nodes are nodes that have been assigned a guess, which was then confirmed to be correct.

\myskip
\mypar{Cuts and isles.}  Fix a view $\mathcal{V}$ of a history tree $\mathcal H$. A set of nodes $C$ in $\mathcal{V}$ is said to be a \emph{cut} for a node $v\notin C$ of $\mathcal{V}$ if two conditions hold: (i)~for every leaf $v'$ of $\mathcal{V}$ that lies in the subtree hanging from $v$, the black path from $v$ to $v'$ contains a node of $C$, and (ii)~no proper subset of $C$ satisfies condition~(i). A cut for the root $r$ whose nodes are all counted is said to be a \emph{counting cut} (see~\cref{xfig:mix}, right).

Let $s$ be a counted node in $\mathcal{V}$, and let $F$ be a cut for $v$ whose nodes are all counted. Then, the set of nodes spanned by the black paths from $s$ to the nodes of $F$ is called \emph{isle}; $s$ is the \emph{root} of the isle, while each node in $F$ is a \emph{leaf} of the isle (see~\cref{xfig:mix}, right). The nodes in an isle other than the root and the leaves are called \emph{internal}. An isle is said to be \emph{trivial} if it has no internal nodes.

If $s$ is an isle's root and $F$ is its set of leaves, we have $a(s)\geq \sum_{v\in F} a(v)$, because $s$ may have some descendants in the history tree $\mathcal H$ that do not appear in the view $\mathcal{V}$. If equality holds, then the isle is said to be \emph{complete}; in this case, we can easily compute the anonymities of all the internal nodes by adding up anonymities starting from the nodes in $F$ and working our way up to $s$.

\myskip
\mypar{Algorithm overview.} Our counting algorithm repeatedly assigns guesses to nodes based on known anonymities (starting from the nodes corresponding to the leader). Eventually some nodes become heavy, and the criterion in \cref{xl:crit} causes the deepest of them to become counted. In turn, counted nodes eventually form isles; the internal nodes of complete isles are marked as counted, which gives rise to more guessers, and so on. In the end, if a counting cut has been created, a simple condition determines whether the anonymities of its nodes add up to the correct count, $n$.

\myskip
\mypar{Algorithm details.} Our complete counting algorithm is found in \cref{xl:algorithm2}. The algorithm takes as input a view $\mathcal V$ (which, we recall, is the history of a process) and uses flags to mark nodes as ``guessed'' or ``counted''; initially, no node is marked. Thanks to these flags, we can check if a node $u\in \mathcal V$ is a guesser: let $u_1$, $u_2$, \dots, $u_k$ be the children of $u$ that are also in $\mathcal V$ (recall that a view does not contain all nodes of a history tree); $u$ is a \emph{guesser} if and only if it is marked as counted, all the $u_i$'s are marked as counted, and $a(u)=\sum_i a(u_i)$.

The algorithm will ensure that nodes marked as guessed are well-spread at all times; if a level of $\mathcal V$ contains a guessed node, it is said to be \emph{locked}. A level $L_t$ is \emph{guessable} if it is not locked and has a non-counted node $v$ that is guessable, i.e., there is a guesser $u$ in $L_{t-1}$ and the red edge $\{v,u\}$ is present in $\mathcal V$ with positive multiplicity.

The algorithm starts by assigning an anonymity of~$1$ to all leader nodes, marking them as counted. Then, as long as there are guessable levels, it keeps assigning guesses to non-counted nodes. When a guess is made on a node $v$,  some nodes in the path from $v$ to the root may become heavy; if so, the algorithm marks the deepest heavy node $v'$ as counted. Furthermore, if the newly counted node $v'$ is the root or a leaf of a complete isle $I$, then the anonymities of all the internal nodes of $I$ are determined, and such nodes are marked as counted (this also unlocks their levels if such nodes were marked as guessed).

Finally, when there are no more guessable levels, the algorithm checks if the termination condition is satisfied, as follows. Each counting cut yields an estimate of $n$, which is given by the sum of the anonymities of its nodes. The terminating condition is satisfied if and only if there is a counting cut $C$ whose total anonymity $n'$ is not greater than the difference between the current round $t'$ and the round $t$ corresponding to the deepest node of $C$ (note that $t'$ can be inferred from the height of the view $\mathcal V$). If so, $n'$ is guaranteed to be equal to the correct number of processes $n$, and the anonymities of the nodes in $L_0$ can be easily computed from $C$.

\myskip
\mypar{Invariants.} We argue that the above algorithm maintains some \emph{invariants}, i.e., conditions that are satisfied every time Line~9 is reached. Namely, (i)~the nodes marked as guessed are well spread, (ii)~there are no heavy nodes, and (iii)~all complete isles are trivial.

These can be verified by induction: (i)~the algorithm never makes a guess in a locked level; (ii)~as soon as a new guess in Line~10 creates some heavy nodes, the deepest one becomes counted (Line~14), making all other nodes non-heavy (in Lines~16--18, weights may only decrease, and no heavy nodes are created); (iii)~as soon as a complete isle $I$ is created due to a node $v'$ being marked as counted in Line~14, $I$ is immediately reduced to a set of trivial isles (Lines~16--18).

The invariants also imply that Lines~10--14 are correct: if no nodes are heavy to begin with, the new guessed node $v$ may create heavy nodes only on the black path from $v$ to $r$. Thus, the first heavy node along this path has a correct guess due to \cref{xl:crit}. Furthermore, since the anonymities assigned in Line~14 are correct, then so are the ones assigned in Line~17.

\myskip
\mypar{Termination condition.} We will now prove that Lines~19--31 are correct: the algorithm indeed gives the correct output if the termination condition is met. We already know that the anonymities computed for the nodes of the counting cut $C$ are correct, and hence we only have to prove that the set $P$ of processes represented by the nodes of $C$ includes all processes. Assume the contrary; \cref{xl:propamain} implies that, if $t'\geq t+|P|=t+n'$, there is a node $z\in L_t$ representing some process not in $P$. Thus, the black path from $z$ to the root $r$ does not contain any node of $C$, contradicting the fact that $C$ is a counting cut whose deepest node is in $L_t$. So, the termination condition is correct.

\myskip
\mypar{Running time.} We have proved that the algorithm is correct; we will now study its running time.

\begin{lemma}\label{xl:bound1}
Whenever Line~9 is reached, at most $n-1$ levels are locked.
\end{lemma}
\begin{proof}
We will prove that, if the subtree hanging from a node $v$ of $\mathcal V$ contains more than $a(v)$ guessed nodes, then it contains a guessed node $v'$ such that $w(v')>a(v')$. The proof is by well-founded induction based on the subtree relation in $\mathcal V$. If $v$ is guessed, then we can take $v'=v$. Otherwise, by the pigeonhole principle, $v$ has at least one child $u$ whose hanging subtree contains more than $a(u)$ guessed nodes. Thus, $v'$ is found in this subtree by the induction hypothesis.

Assume for a contradiction that at least $n$ levels of $\mathcal V$ are locked; hence, $\mathcal V$ contains at least $n$ guessed nodes. None of the nodes representing the leader is ever guessed, because all of them are marked as counted in Lines~6--7. Hence, all of the guessed nodes must be in the subtrees hanging from the nodes in $L_0$ representing non-leader nodes, whose total anonymity is $n-1$. Thus, by the pigeonhole principle, the subtree of one such node $v\in L_0$ contains more than $a(v)$ guessed nodes. The node $v$ must be in $\mathcal V$, and therefore the subtree hanging from $v$ contains a guessed node $v'$ such that $w(v')>a(v')$.

Since the algorithm's invariant~(i) holds, we can apply \cref{xl:limit} to $v'$, which implies that there exist heavy nodes. In turn, this contradicts invariant~(ii). We conclude that at most $n-1$ levels are locked.
\end{proof}

\begin{lemma}\label{xl:bound2}
Assume that all the levels of the history tree up to $L_t$ are entirely contained in the view $\mathcal V$. Then, whenever Line~9 is reached and a counting cut has not been created yet, there are at most $n-2$ levels in the range from $L_1$ to $L_{t}$ that lack a guessable non-counted node.
\end{lemma}
\begin{proof}
Lines~6--7 create guessers in every level from $L_0$ to $L_{t-1}$ (these are the nodes representing the leader); hence, these levels must have a non-empty set of guessers at all times. Consider any level $L_i$ with $1\leq i\leq t$ such that all the guessable nodes in $L_i$ are already counted. Let $S$ be the set of guessers in $L_{i-1}$; note that not all nodes in $L_{i-1}$ are guessers, or else they would constitute a counting cut. The network is 1-interval-connected, so there is a red edge $\{u,v\}$ (with positive multiplicity) such that $u\in S$ and the parent of $v$ is not in $S$. By definition, the node $v$ is guessable; therefore, it is counted. Also, since the parent of $v$ is not a guesser, $v$ must have a non-counted parent or a non-counted sibling; note that such a non-counted node is in $\mathcal V$.

We have proved that every level (up to $L_{t}$) lacking a guessable non-counted node contains a counted node $v$ having a parent or a sibling that is not counted: we call such a node $v$ a \emph{bad} node. To conclude the proof, it suffices to show that there are at most $n-2$ bad nodes up to $L_{t}$.

We will prove by induction that, if a subtree $\mathcal W$ of $\mathcal V$ contains the root $r$, no counting cuts, and no non-trivial isles, then $\mathcal W$ contains at most $f-1$ bad nodes, where $f$ is the number of leaves of $\mathcal W$ not representing the leader. The base case is $f=1$, which holds because any counted non-leader node in $\mathcal W$ gives rise to a counting cut. For the induction step, let $v$ be a bad node of maximum depth in $\mathcal W$. Let $(v_1, v_2,\dots, v_k)$ be the black path from $v_1=v$ to the root $v_k=r$, and let $1<i\leq k$ be the smallest index such that $v_i$ has more than one child in $\mathcal W$ ($i$ must exist, because $v_k$ branches into leader and non-leader nodes). Let $\mathcal W'$ be the tree obtained by deleting the black edge $\{v_{i-1},v_i\}$ from $\mathcal W$, as well as the subtree hanging from it. Notice that the induction hypothesis applies to $\mathcal W'$: since $v_1$ is counted, and none of the nodes $v_2$, \dots, $v_{i-1}$ are branching, the removal of $\{v_{i-1},v_i\}$ does not create counting cuts or non-trivial isles. Also, $v_2$ is not counted (unless perhaps $v_2=v_i$), because $v_1$ is bad. Furthermore, none of the nodes $v_3$, \dots, $v_{i-1}$ is counted, or else $v_2$ would be an internal node of a (non-trivial) isle. Therefore, $\mathcal W'$ has exactly one less bad node than $\mathcal W$ and one less leaf; the induction hypothesis now implies that $\mathcal W$ contains at most $f-1$ bad nodes.

Observe that the subtree $\mathcal V'$ of $\mathcal V$ formed by all levels up to $L_{t}$ satisfies all of the above conditions, as it contains the root $r$ and has no counting cuts, because a counting cut for $\mathcal V'$ would be a counting cut for $\mathcal V$, as well. Also, invariant~(iii) ensures that $\mathcal V'$ contains no non-trivial complete isles. However, since the levels up to $L_{t}$ are contained in $\mathcal V'$, all isles in $\mathcal V'$ are complete, and thus must be trivial. We conclude that, if $\mathcal V'$ has $f$ non-leader leaves, it contains at most $f-1$ bad nodes. Since the non-leader leaves of $\mathcal V$ induce a partition of the $n-1$ non-leader processes, we have $f\leq n-1$, implying that the number of bad nodes up to $L_{t}$ is at most $n-2$.
\end{proof}

\begin{theorem}\label{xth:term}
There is an algorithm that solves the Generalized Counting problem in 1-interval-connected anonymous dynamic networks with a leader and terminates in at most $3n-2$ rounds.
\end{theorem}
\begin{proof}
It suffices to prove that, if a process $p$ executes the above algorithm at round~$t'=3n-2$, it terminates. By \cref{xl:propa}, the levels of the history tree up to $L_{2n-1}$ completely appear in $p$'s history. Due to \cref{xl:bound1,xl:bound2}, as long as a counting cut has not been created, at least one level between $L_1$ and $L_{2n-2}$ is guessable (because at most $n-1$ levels are locked and at most $n-2$ levels lack a guessable non-counted node). Hence, when Line~19 is reached, a counting cut $C$ has been found whose deepest node is in $L_t$, with $t\leq 2n-2$. The level $L_t$ is completely contained in $p$'s history, and therefore $n'=\sum_{v\in C}a(v)=n$. Thus, $t'=3n-2\geq t+n'$, and $p$ terminates.
\end{proof}

Our analysis of the algorithm is tight: In \cref{as:3.2}, we will show that there are 1-interval-connected networks where the algorithm terminates in exactly $3n-3$ rounds.

\section{Negative Results}\label{xs:5}
\mypar{Unsolvable Problems.}
We will now prove that the multi-aggregation problems introduced in \cref{xs:2} are the only problems that can be solved deterministically in anonymous networks.

\begin{theorem}\label{xth:imp}
No problem other than the multi-aggregation problems can be solved deterministically, even when restricted to simple connected static networks with a unique leader.
\end{theorem}
\begin{proof}
Let us consider the (static) network whose topology at round~$t$ is the complete graph $G_t=K_n$, i.e., each process receives messages from all other processes at every round. We can prove by induction that all nodes of the history tree other than the root have exactly one child. This is because any two processes with the same input always receive equal multisets of messages, and are therefore always indistinguishable. Thus, the history tree is completely determined by the multiset $\mu$ of all processes' inputs; moreover, a process' history at any given round only depends on the process' own input and on $\mu$. By \cref{xth:view}, this is enough to conclude that if a process' output stabilizes, that output must be a function of the process' own input and of $\mu$, which is the defining condition of a multi-aggregation problem.
\end{proof}

\mypar{Lower Bound.} We now prove a lower bound of roughly $2n$ rounds on the Counting problem. In \cref{as:3.3}, we will provide additional (albeit weaker) lower bounds for several other problems.

We first introduce a family of 1-interval-connected dynamic networks. For any $n\geq 1$, we consider the dynamic network $\mathcal G_n$ whose topology at round~$t$ is the graph $G^{(n)}_t$ defined on the system $\{p_1,p_2,\dots,p_n\}$ as follows. If $t\geq n-2$, then $G^{(n)}_t$ is the path graph $P_n$ spanning all processes $p_1$, $p_2$, \dots, $p_n$ in order. If $1\leq t\leq n-3$, then $G^{(n)}_t$ is $P_n$ with the addition of the single edge $\{p_{t+1}, p_n\}$. We assume $p_1$ to be the leader and all other processes to have the same input.

\begin{theorem}\label{xth:lower}
No deterministic algorithm can solve the Counting problem in less than $2n-6$ rounds, even  in a simple 1-interval-connected dynamic network, and even if there is a unique leader in the system. If termination is required, the bound improves to $2n-4$ rounds.\footnote{A slightly better bound can be obtained if we allow networks with double links. If we double the edge $\{p_{n-1}, p_n\}$ in $G^{(n)}_{n-2}$ and we add a double self-loop on $p_n$ in $G^{(n)}_t$ for all $t\geq n-1$, we obtain a lower bound of $2n-3$ rounds for stabilization and $2n-1$ rounds for termination.}
\end{theorem}
\begin{proof}
Let us consider the network $\mathcal G_n$ as defined above. It is straightforward to prove by induction that, at round~$t\leq n-3$, the process $p_{t+1}$ gets disambiguated, while all processes $p_{t+2}$, $p_{t+3}$, \dots, $p_n$ are still indistinguishable. So, the history tree of $\mathcal G_n$ has a very regular structure, which is illustrated in \cref{xfig:lower1,xfig:lower2}. By comparing the history trees of $\mathcal G_n$ and $\mathcal G_{n+1}$, we see that the leaders of the two systems have identical histories from round~$1$ up to round~$2n-5$. Thus, by \cref{xth:view}, both leaders must have equal internal states and give equal outputs up to round~$2n-5$.

It follows that, if the leader of $\mathcal G_n$ were to output the number $n$ and terminate in less than $2n-4$ rounds, then the leader of $\mathcal G_{n+1}$ would do the same, producing the incorrect output $n$ instead of $n+1$. Assume now that the leaders of $\mathcal G_n$ and $\mathcal G_{n+1}$ could stabilize on the correct output in less than $f(n)=2n-6$ and $f(n+1)=2(n+1)-6=2n-4$ rounds, respectively. Then, at round~$2n-5$, the two leaders would output $n$ and $n+1$ respectively, contradicting the fact that they should give the same output.
\end{proof}

\mypar{Acknowledgments.} The authors would like to thank Gregory Schwartzman for useful comments.

\section{Conclusion} 
We have presented an algorithm for the Generalized Counting problem that terminates in $3n-2$ rounds, which allows the linear-time computation of all deterministically computable functions in 1-interval-connected anonymous dynamic networks with a leader.

We also proved a lower bound of roughly $2n$ rounds for the Counting problem, both for terminating and for stabilizing algorithms, and we gave a stabilizing algorithm for Generalized Counting that matches the lower bound (up to a small additive constant). Determining whether the $2n$ lower bound is tight for terminating algorithms is left as an open problem.

We introduced the novel concept of \emph{history tree} as our main investigation technique. History trees are a powerful tool that completely and naturally captures the concept of symmetry in anonymous dynamic networks. We have demonstrated the effectiveness of our methods by optimally solving a wide class of fundamental problems, and we argue that our technique can be used in similar settings, such as the case with multiple leaders, as well as to obtain better bounds for networks with special topologies or randomly generated networks.

\clearpage
\begin{figure}
\centering
\includegraphics[scale=0.575]{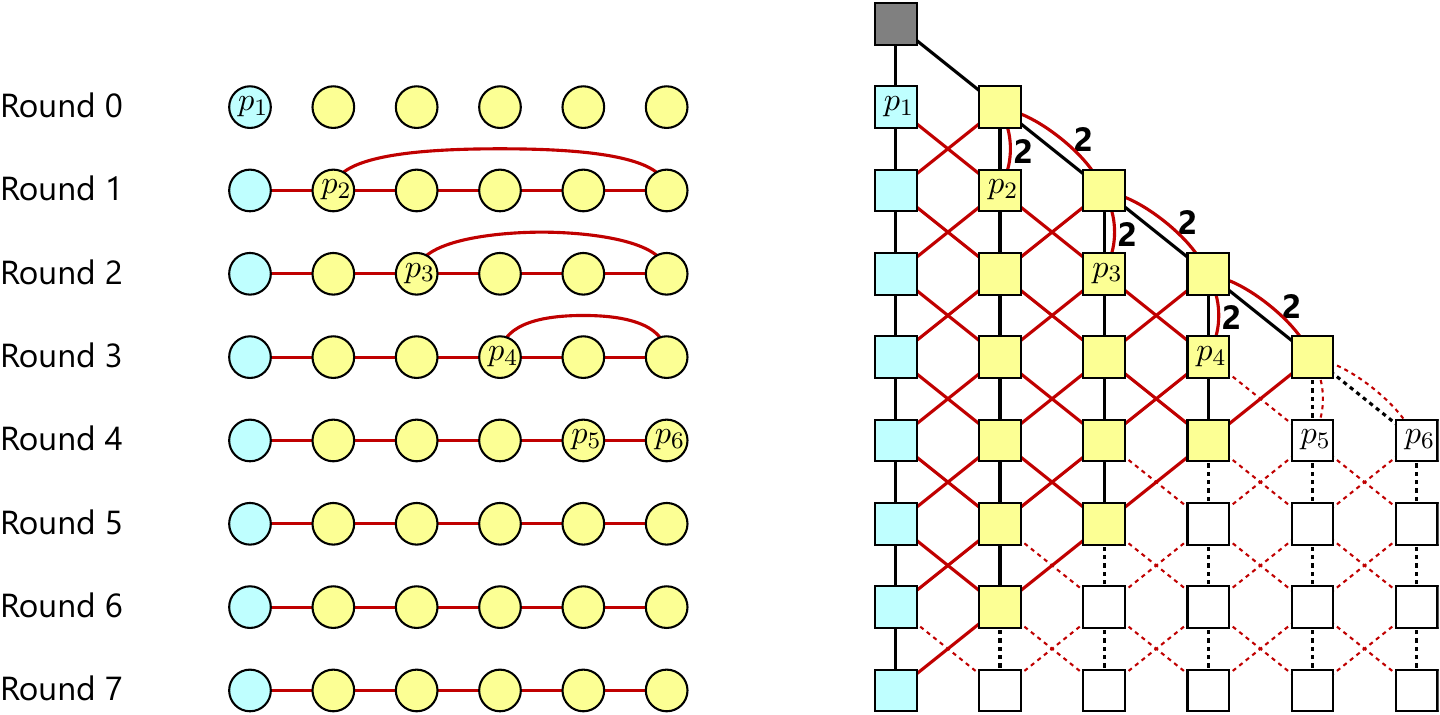}
\caption{The first rounds of the dynamic network $\mathcal G_n$ used in \cref{xth:lower} (left) and the corresponding levels of its history tree (right), where $n=6$; the process in blue is the leader. The white nodes and the dashed edges in the history tree are not in the history of the leader at round~$7$. The labels $p_1$, \dots, $p_6$ have been added for the reader's convenience, and mark the processes that get disambiguated, as well as their corresponding nodes of the history tree, which have anonymity~$1$.}
\label{xfig:lower1}
\end{figure}

\begin{figure}
\centering
\includegraphics[scale=0.575]{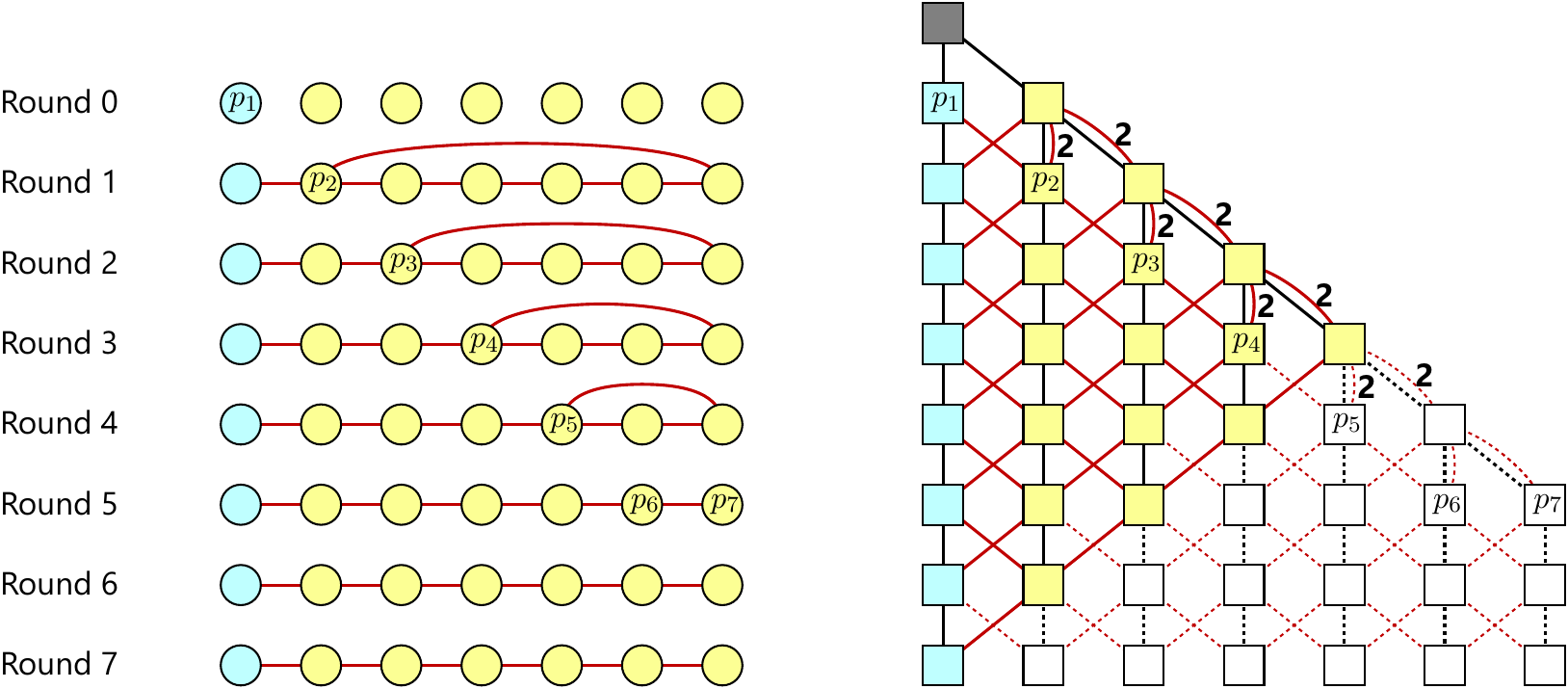}
\caption{The first rounds of the dynamic network $\mathcal G_{n+1}$ with $n=6$. Observe that the history of the leader at round~$7$ is identical to the history highlighted in \cref{xfig:lower1}. The intuitive reason is that, from round~$1$ to round~$n-3$, both networks have a cycle whose processes are all indistinguishable (and are therefore represented by a single node in the history tree), except for the one process with degree~$3$. Thus, the history trees of $\mathcal G_n$ and $\mathcal G_{n+1}$ are identical up to level~$n-3$. After that, the two networks get disambiguated, but this information takes another $n-3$ rounds to reach the leader. Therefore, if the leader of $\mathcal G_n$ and the leader of $\mathcal G_{n+1}$ execute the same algorithm, they must have the same internal state up to round~$2n-5$, due to \cref{xth:view}. In particular, they cannot give different outputs up to that round, which leads to our lower bounds on stabilization and termination for the Counting problem.}
\label{xfig:lower2}
\end{figure}

\clearpage

\begin{center}
\begin{LARGE}
\textbf{APPENDIX}
\end{LARGE}
\end{center}

\appendix

\section{Formal Model Definition and Basic Results}\label{as:1}
A \emph{multiset} on an \emph{underlying set} $X$ is a function $\mu\colon X\to \mathbb N$. The non-negative integer $\mu(x)$ is the \emph{multiplicity} of $x\in X$, and specifies how many \emph{copies} of each element of $X$ are in the multiset. The set of all multisets on the underlying set $X$ is denoted as $\mult{X}$.

\subsection{Model of Computation}\label{as:1.1}
A \emph{dynamic network} is an infinite sequence $\mathcal G=(G_i)_{i\geq 1}$, where $G_i=(V,E_i)$ is an undirected multigraph, i.e., $E_i$ is a multiset of unordered pairs of elements of $V$. In this context, the set $V=\{1,2,\dots, n\}$ is called \emph{system}, and its $n\geq 1$ elements are the \emph{processes}. The elements of the multiset $E_i$ are called \emph{links}; note that we allow any number of ``parallel links'' between two processes.\footnote{In the dynamic networks literature, $G_i$ is typically assumed to be a simple graph, as at most one link between the same two processes is allowed. However, our results hold more generally for multigraphs.}

The standard model of computation for systems of processes specifies the following \emph{computation parameters}:
\begin{itemize}
\item three sets $\mathcal I$, $\mathcal O$, $\mathcal S$, representing the possible \emph{inputs}, \emph{outputs}, and \emph{internal states} for a process, respectively;
\item A partition of $\mathcal S$ into two subsets $\mathcal S_T$ and $\mathcal S_N$, representing the \emph{terminal} and \emph{non-terminal} states, respectively;
\item an \emph{input map} $\iota\colon \mathcal I\to\mathcal S$, where $\iota(x)$ represents the initial state of a process whose input is $x$;
\item an \emph{output map} $\omega\colon \mathcal S_T\to\mathcal O$, where $\omega(s)$ represents the output of a process in a terminal state $s$;
\item a function $\mathcal A\colon \mathcal S_N\times \mult{\mathcal S}\to \mathcal S$, representing a \emph{deterministic algorithm} for local computations. The algorithm takes as input the (non-terminal) state of a process $p$, as well as the multiset of states of all processes linked to $p$, and it outputs the new state for $p$.
\end{itemize}
Given the above parameters and a dynamic network $\mathcal G=(G_i)_{i\geq 1}$, the computation proceeds as follows. The system $V$ is assigned an \emph{input} in the form of a function $\lambda\colon V\to \mathcal I$. Time is discretized into units called \emph{rounds} $r_1$, $r_2$, $r_3$, etc.; the multigraph $G_i$ describes the links that are present at round $r_i$, and is called the \emph{topology} of the network at round $r_i$.

Each process in the system updates its own state at the end of every round. We denote by $\sigma_i(p)\in \mathcal S$ the state of process $p\in V$ after round $r_i$, for $i\geq 1$. Moreover, we define the \emph{initial state} of $p$ as $\sigma_0(p)=\iota(\lambda(p))$ (we could say that each process is assigned its initial state at ``round $r_0$'').

During round $r_i$, each process $p\in V$ \emph{broadcasts} its state to all its neighbors in $G_i$; then, $p$ receives the multiset $M_i(p)$ of states of all its neighbors (one state per incident link in $E_i$).\footnote{In a slightly different model, a process does not necessarily broadcast its entire state, but a \emph{message}, which in turn is computed as a function of the internal state.} Finally, if $p$'s state is non-terminal, it computes its next state based on its current state and $M_i(p)$. In formulas, if $\sigma_{i-1}(p)\in \mathcal S_T$, then $\sigma_i(p)=\sigma_{i-1}(p)$. Otherwise, $\sigma_i(p)=\mathcal A(\sigma_{i-1}(p),M_i(p))$, where $M_i(p)$ is defined as
\[M_i(p)\colon \mathcal S\to \mathbb N\quad \mbox{such that}\quad \mathcal S \ni s \mapsto \sum_{\substack{q\in V \\ \sigma_{i-1}(q)=s}}E_i(\{p,q\}).\footnote{Recall that $G_i$ is a multigraph, and therefore $E_i$ is a multiset, i.e., a function that maps every possible edge to its multiplicity (if an edge is absent, its multiplicity is $0$).}\]
Note that processes are \emph{anonymous}: they are initially identical and indistinguishable except for their input, and they all update their internal states by executing the same deterministic algorithm $\mathcal A$.

Once all processes in $V$ are in a terminal state, the system's \emph{output} is defined as the function $\zeta\colon V\to \mathcal O$ such that $\zeta(p)=\omega(\sigma_t(p))$. If this happens for the first time after round $r_t$, we say that the execution \emph{terminates} in $t$ rounds.

Given an input set $\mathcal I$ and an output set $\mathcal O$, a \emph{problem} on $(\mathcal I, \mathcal O)$ is a sequence of functions $P=(\pi_n)_{n\geq 1}$, where $\pi_n$ maps every function $\lambda\colon \{1,2,\dots, n\} \to \mathcal I$ to a function $\zeta\colon \{1,2,\dots, n\} \to \mathcal O$. Essentially, a problem prescribes a relationship between inputs and outputs: whenever a system of $n$ processes is assigned a certain input $\lambda$, it must eventually terminate with the output $\zeta=\pi_n(\lambda)$. We denote the set of all problems on $(\mathcal I, \mathcal O)$ as $\mathcal P(\mathcal I, \mathcal O)$.

We say that a problem $P=(\pi_n)_{n\geq 1}\in \mathcal P(\mathcal I, \mathcal O)$ is \emph{solvable in $f(n)$ rounds} if there exists computation parameters (i.e., an algorithm $\mathcal A$, as well as a set of states $\mathcal S=\mathcal S_N\cup\mathcal S_T$, an input map $\iota$, and an output map $\omega$) such that, whenever a system $V=\{1,2,\dots, n\}$ of $n$ processes is given an input $\lambda\colon V\to \mathcal I$ and carries out its computation on a dynamic network $\mathcal G$ as described above, it terminates in at most $f(n)$ rounds and outputs $\pi_n(\lambda)$, regardless of the topology of $\mathcal G$. Note that the algorithm $\mathcal A$ must be the same for every $n$, i.e., it is \emph{uniform}; in other words, the system is unaware of its own size.

A weaker notion of solvability involves \emph{stabilization} instead of termination. Here, the processes are only required to output $\pi_n(\lambda)$ starting at round $f(n)$ and at all subsequent rounds, without necessarily reaching a terminal state.

Since only trivial problems can be solved if no restrictions are made on the topology of the dynamic network $\mathcal G$ (think of a dynamic network with no links at all), it is customary to require that the dynamic network be \emph{1-interval-connected}. That is, if $\mathcal G=(G_i)_{i\geq 1}$, we assume that $G_i$ is a connected multigraph for all $i\geq 1$.

Another assumption that is often made about the system is the presence of a \emph{unique leader} among the processes. That is, the input set $\mathcal I$ is of the form $\mathcal I=\{L,N\}\times \mathcal I'$, and an input assignment $\lambda\colon V\to \mathcal I$ is valid if and only if there is exactly one process $p\in V$, the \emph{leader}, such that $\lambda(p)=(L, x)$ for some $x\in \mathcal I'$ (thus, all non-leader processes have an input of the form $(N,x)$).

\subsection{Multi-Aggregation Problems}\label{as:1.2}
In this section we will define an important class of problems in $\mathcal P(\mathcal I, \mathcal O)$, the \emph{multi-aggregation problems}. As we will see, these are precisely the problems that can be solved in 1-interval-connected anonymous dynamic networks with a unique leader.

Given a system $V$ and an input assignment $\lambda\colon V\to \mathcal I$, we define the \emph{inventory} of $\lambda$ as the function $\mu_\lambda\colon \mathcal I\to \mathbb N$ which counts the processes that are assigned each given input. Formally, $\mu_\lambda$ is the multiset on $\mathcal I$ such that $\mu_\lambda(x) = |\lambda^{-1}(x)|$ for all $x\in\mathcal I$.

A problem $P\in \mathcal P(\mathcal I, \mathcal O)$ is said to be a \emph{multi-aggregation problem} if the output to be computed by each process only depends on the process' own input and the inventory of all processes' inputs. Formally, $P=(\pi_n)_{n\geq 1}$ is a multi-aggregation problem if there is a function $\psi\colon \mathcal I \times \mult{\mathcal I} \to \mathcal O$, called \emph{signature function}, such that every input assignment $\lambda\colon \{1,2,\dots, n\} \to \mathcal I$ is mapped by $\pi_n$ to the output $\zeta\colon \{1,2,\dots, n\}\to \mathcal O$ where $\zeta(p)=\psi(\lambda(p), \mu_\lambda)$.

\subsection{Counting Problem}\label{as:1.3}
If the signature function $\psi(x,\mu)$ of a multi-aggregation problem $P\in\mathcal P(\mathcal I, \mathcal O)$ does not depend on the argument $x$, then $P$ is simply said to be an \emph{aggregation problem}. Thus, all aggregation problems are \emph{consensus problems}, where all processes must terminate with the same output. Notable examples of aggregation problems in $\mathcal P(\mathbb R, \mathbb R)$ include computing statistical functions on the input values, such as sum, average, maximum, median, mode, variance, etc.

The \emph{Counting Problem} $C_\mathcal I$ is the aggregation problem in $\mathcal P(\mathcal I, \mult{\mathcal I})$ whose signature function is $\psi(x,\mu)=\mu$. That is, all processes must determine the inventory of their input assignment $\lambda$, i.e., count the number of processes that have each input. If all inputs are the same (or, in the presence of a leader, if $\mathcal I =\{L,N\}$), the Counting Problem reduces to determining the total number of processes in the system, $n$.

The Counting Problem is \emph{complete} for the class of multi-aggregation problems, in the sense that solving it efficiently implies solving any other multi-aggregation problem efficiently (actually, with no overhead at all). The following theorem makes this observation precise.

\begin{theorem}\label{th:compl}
For every input set $\mathcal I$ and output set $\mathcal O$, if the Counting Problem $C_\mathcal I$ is solvable in $f(n)$ rounds, then every multi-aggregation problem in $\mathcal P(\mathcal I, \mathcal O)$ is solvable in $f(n)$ rounds, as well.
\end{theorem}
\begin{proof}
The proof is essentially the same whether the requirement is termination or stabilization. For brevity, we only discuss termination.

Let $\mathcal A$ be an algorithm that solves $C_\mathcal I$ in $f(n)$ rounds with the computation parameters $\mathcal S=\mathcal S_N\cup \mathcal S_T$, $\iota$, $\omega$ as defined in \cref{as:1.1}. We will show how to solve $P$ by modifying $\mathcal A$ and the above parameters.

The idea is that each process should execute $\mathcal A$ while remembering its own input $x$. Once it has computed the inventory $\mu_\lambda$ of the system's input assignment, it immediately uses $x$ and $\mu_\lambda$ to compute the signature function of $P$.

We define the set of internal states for $P$ as $\mathcal S'=\mathcal S\times\mathcal I$, with non-terminal states $\mathcal S'_N=\mathcal S_N\times\mathcal I$. The new input map is $\iota'(x)=(\iota(x),x)\in \mathcal S'$ for all $x\in\mathcal I$. The algorithm $\mathcal A'\colon \mathcal S'_N\times \mult{\mathcal S'}\to \mathcal S'$ is defined as $\mathcal A'((s,x),\mu)=(\mathcal A(s,\mu'),x)$, where $\mu'$ is defined as
\[\mu'\colon \mathcal S\to \mathbb N\quad \mbox{such that}\quad \mathcal S\ni s \mapsto \sum_{x\in \mathcal I}\mu((s,x)).\]
Finally, we define the output map as $\omega'((s,x))=\psi(x,\omega(s))$, where $\psi\colon \mathcal I \times \mult{\mathcal I} \to \mathcal O$ is the signature function of $P$.

Since $\mathcal A$ solves the Counting Problem, when a system $V$ of $n$ processes is assigned an input $\lambda$ and executes $\mathcal A$ on a network $\mathcal G$, within $f(n)$ rounds each process $p\in V$ reaches a terminal state $s\in \mathcal S_T$ such that $\omega(s)=\mu_\lambda$. Therefore, if the system executes $\mathcal A'$ on the same input and in the same network, the process $p$ reaches the terminal state $(s,\lambda(p))\in\mathcal S'$ in the same number of rounds. Thus, $p$ gives the output $\omega'((s,\lambda(p)))=\psi(\lambda(p),\omega(s)) = \psi(\lambda(p),\mu_\lambda)$, as required by the problem $P$.
\end{proof}

We remark that \cref{th:compl} makes no assumption on the dynamic network's topology or the presence of a leader.

\section{Formal Definition of History Trees and Basic Properties}\label{as:2}
\subsection{Abstract Structure of a History Tree}\label{as:2.1}

A \emph{history tree} on an input set $\mathcal I$ is a quintuplet $\mathcal H=(H, u, B, R, \ell)$ such that:\footnote{Note that, in \cref{xfig:hg}, the root node $u$ has been renamed $r$ to match the notation of \cref{xs:2}.}
\begin{itemize}
\item $H$ is a countably infinite set of \emph{nodes}, with $u\in H$;
\item $(H,B)$ is an infinite undirected tree rooted at $u$, where every node has at least one child;
\item $(H,R)$ is an infinite undirected multigraph, where all edges have finite multiplicities;
\item $\ell\colon H\setminus\{u\}\to \mathcal I$ is a function that assigns a \emph{label} $\ell(h)\in\mathcal I$ to each node $h\in H$ other than $u$.
\end{itemize}
The elements of $B$ are called \emph{black edges}, and $(H,B)$ is the \emph{black tree} of $\mathcal H$. Similarly, $(H,R)$ is the \emph{red multigraph} of $\mathcal H$, and its edges (with positive multiplicity) are called \emph{red edges}.

The \emph{depth} of $h\in H$ is the distance between the nodes $u$ and $h$ as measured in the black tree. For all $i\geq -1$, we define the $i$th \emph{level} $L_i\subseteq H$ as the set of nodes that have depth $i+1$. Thus, $L_{-1}=\{u\}$.

The nodes of a history tree inherit their ``parent-child-sibling'' relationships from the black tree: for example, $u$ is the \emph{parent} of all the nodes in $L_0$ because it is connected to all of them by black edges; thus, every two nodes in $L_0$ are \emph{siblings}, etc.

A \emph{descending path} in a history tree is a sequence of $k\geq 1$ nodes $(h_1, h_2, \dots, h_k)$ such that, for all $1\leq i<k$, the unordered pair $\{h_i,h_{i+1}\}$ is either a black or a red edge and, if $h_i\in L_j$, then $h_{i+1}\in L_{j'}$ with $j'>j$. Equivalently, we say that $(h_k, h_{k-1}, \dots, h_1)$ is an \emph{ascending path}. The node $h_1$ is the \emph{upper endpoint} of the path, and $h_k$ is the \emph{lower endpoint}. If $k=1$, then $h_1$ is both the upper and the lower endpoint.

Let $h$ be a node in a history tree $\mathcal H=(H, u, B, R, \ell)$, and let $H_h\subset H$ be the set of all upper endpoints of ascending paths whose lower endpoint is $h$. The \emph{view} of $h$ in $\mathcal H$, denoted as $\nu(h)$, is the (finite) subtree of $\mathcal H$ induced by $H_h$. Namely, $\nu(h) = (H_h, B_h, R_h, \ell_h)$, where $B_h\subset B$ is the set of black edges whose endpoints are both in $H_h$, $R_h$ is the restriction of $R$ to the unordered pairs of nodes in $H_h$, and $\ell_h$ is the restriction of $\ell$ to $H_h$. The node $h$ is said to be the \emph{viewpoint} of the view $\nu(h)$; note that the viewpoint is the unique deepest node in the view. We will denote by $\View{\mathcal I}$ the set of all views of all history trees on the input set $\mathcal I$.

\subsection{History Tree of a Dynamic Network}\label{as:2.2}

Given a dynamic network $\mathcal G=(G_i=(V,E_i))_{i\geq 1}$ and an input assignment $\lambda\colon V\to\mathcal I$, we will show how to construct a history tree $\mathcal H=(H, u, B, R, \ell)$ with labels in $\mathcal I$ that is naturally associated with $\mathcal G$ and $\lambda$.

At the same time, for all $i\geq -1$, we will construct a \emph{representation function} $\rho_i\colon V\to L_i\subset H$. Intuitively, if $h\in L_i$, the preimage $\rho_i^{-1}(h)$ is a set of processes that, based on their ``history'', are necessarily ``indistinguishable'' at the end of round $r_i$ (we will also give a meaning to this sentence when $i=-1$ and $i=0$). The \emph{anonymity} of a node $h\in L_i\subset H$ is the number of processes that $h$ represents, and is given by the function $\alpha\colon H\to \mathbb N^+$, where $\alpha(h)=|\rho_i^{-1}(h)|$.

In this paradigm, the label $\ell(h)$ of a node $h\neq u$ is the input $\lambda(p)$ that each process $p$ represented by $h$ has received at the beginning of the execution (all such processes must have received the same input, or else they would have different histories, and they would not be \emph{necessarily} indistinguishable).

The black tree keeps track of the progressive ``disambiguation'' of processes: if two processes are represented by the same node $h\in L_{i-1}$, they have had the same history up to round $r_{i-1}$. However, if they receive different multisets of messages at round $r_i$, they are no longer \emph{necessarily} indistinguishable, and will therefore be represented by two different nodes in $L_i$, each of which is a child of $h$ in $\mathcal H$.

Red edges represent ``observations'': if, at round $r_i$, each of the processes represented by node $h\in L_i$ has received messages from processes represented by node $h'\in L_{i-1}$ through a total of $m$ links in $G_i$, then $R$ contains the red edge $\{h,h'\}$ with multiplicity $m$.

We inductively construct the history tree $\mathcal H$ and the representation functions $\rho_i$ level by level. First we define $\rho_{-1}(p)=u$ for all $p\in V$, where $u$ is the root of $\mathcal H$. The intuitive meaning is that, before processes are assigned inputs (i.e., at ``round $r_{-1}$''), they are all indistinguishable, because the network is anonymous. Thus, the anonymity of the root node $u$ is $|V|=n$.

The level $L_0$ of $\mathcal H$ represents the system at the end of ``round $r_0$'', i.e., after every process $p\in V$ has been assigned its input $\lambda(p)\in\mathcal I$ and has acquired an initial state $\iota(\lambda(p))$. At this point, processes with the same input are necessarily indistinguishable. Thus, for every input $x\in\lambda(V)$, there is a node $h_x\in L_0$ with label $\ell(h_x)=x$. Accordingly, for every process $p\in\lambda^{-1}(x)$, we define $\rho_0(p)=h_x$.

In order to inductively construct the level $L_i$ of $\mathcal H$ for $i\geq 1$, we define the concept of \emph{observation multiset} $o_i(p)\in \mult{L_{i-1}}$ of a process $p\in V$ at round $r_i$. This corresponds to the multiset of ``necessarily indistinguishable'' messages received by $p$ at round $r_i$, and its underlying set is $L_{i-1}$, i.e., the collection of equivalence classes of processes that are necessarily indistinguishable after round $r_{i-1}$. The definition of $o_i(p)$ is similar to the definition of $M_i(p)$ in \cref{as:1.1}, as these are two closely related concepts:
\[o_i(p)\colon L_{i-1}\to \mathbb N\quad \mbox{such that}\quad L_{i-1}\ni h \mapsto \sum_{q\in \rho_{i-1}^{-1}(h)}E_i(\{p,q\}).\]

Now, the children of $h\in L_{i-1}$ in $L_i$ are constructed as follows. Define the equivalence relation $\sim_h$ on the set of processes $V_h=\rho_{i-1}^{-1}(h)$ such that $p\sim_h q$ if and only if $o_i(p)=o_i(q)$. Let $W_1$, $W_2$, \dots, $W_k$ be the equivalence classes of $\sim_h$, with $W_1\cup W_2\cup\dots\cup W_k = V_h$. The node $h$ has exactly $k$ children $h_1$, $h_2$, \dots $h_k$ in $L_i$, one for each equivalence class of $\sim_h$. Thus, for every $1\leq j\leq k$ and every process $p\in W_j$, we define $\rho_i(p)=h_j$. Also, since $W_j\subseteq V_h$ and a process' input never changes, we set $\ell(h_j)=\ell(h)$. The red edges connecting $h_j$ with nodes in $L_{i-1}$ match the observation multiset of the processes in $W_j$. That is, if the node $h'\in L_{i-1}$ has multiplicity $m$ in $o_i(p)$, where $p\in W_j$, then the red edge $\{h_j,h'\}$ has multiplicity $m$ in $R$.

\subsection{Basic Properties of History Trees}\label{as:2.3}
The following properties of history trees are easily derived from the definitions in \cref{as:2.2}.
\begin{observation}\label{obs:hist}
\phantom{}
\begin{itemize}
\item No two nodes in $L_0$ have the same label, and each node in $H\setminus(L_{-1}\cup L_0)$ has the same label as its parent.
\item Edges only connect nodes in adjacent levels; if $\{h,h'\}$ is a black or red edge with $h\in L_i$, then $h'\in L_{i-1}\cup L_{i+1}$.
\item If $h\in L_i$ with $i\geq -1$ and $h_1,h_2,\dots,h_k\in L_{i+1}$ are the children of $h$, then
\[\bigsqcup_{j=1}^k\rho_{i+1}^{-1}(h_j)=\rho_i^{-1}(h),\quad\mbox{and therefore}\quad\sum_{j=1}^k\alpha(h_j) =\alpha(h).\]
\end{itemize}
\end{observation}

We will now give a concrete meaning to the idea that the processes represented by a node of a history tree are ``necessarily indistinguishable''. That is, if a node is in the $i$th level, then all the processes it represents must have the same state at the end of round $r_i$, regardless of the deterministic algorithm being executed (cf.~\cref{cor:same}).

We will first prove a fundamental result: the internal state $\sigma_i(p)$ of any process $p\in V$ at the end of any round $r_i$, with $i\geq 0$, can be inferred from the process' \emph{history} $\hist{i}{p}$, which is defined as the view of the node of $L_i$ representing $p$, i.e., $\hist{i}{p}=\nu(\rho_i(p))$. 

\begin{theorem}\label{th:view}
Given any set of computation parameters (as defined in \cref{as:1.1}), there exists a function $\mathcal F\colon \View{\mathcal I}\to \mathcal S$ such that, for every history tree $\mathcal H$ associated with a dynamic network $\mathcal G$ and an input assignment $\lambda$ (as defined in \cref{as:2.2}), and for every process $p\in V$ and every $i\geq 0$, $\mathcal F(\hist{i}{p})=\sigma_i(p)$.
\end{theorem}
\begin{proof}
Let $h=\rho_i(p)$, and let $\hist{i}{p}=\nu(h) = (H_h, B_h, R_h, \ell_h)$ be the view of $h$ in the history tree $\mathcal H=(H, u, B, R, \ell)$. We will define $\mathcal F$ (in terms of the computation parameters $\iota\colon\mathcal I\to \mathcal S$ and $\mathcal A\colon \mathcal S_N\times \mult{\mathcal S}\to \mathcal S$) in such a way that $\mathcal F(\hist{i}{p})=\mathcal F(\nu(h))=\sigma_i(p)$.

Note that $\mathcal F$ can identify $h$ as the deepest node of $\hist{i}{p}$. Also, it can compute $i$ as the depth of $h$ in the black tree $(H_h, B_h)$ minus $1$. Hence, $h$ and $i$ do not have to be explicitly provided as arguments to $\mathcal F$.

The construction of $\mathcal F$ is done by induction on $i$. If $i=0$, the view $\hist{0}{p}=\nu(h)$ contains the node $h\in L_0$, whose label $\ell_h(h)=\ell(h)$ is the input of $p$, i.e., $\lambda(p)$ (cf.~the construction of $L_0$ in \cref{as:2.2}). Therefore, we set $\mathcal F(\hist{i}{p})=\iota(\ell(h))$; indeed, $\mathcal F(\hist{i}{p}) = \iota(\lambda(p)) = \sigma_0(p)$ (cf.~the definition of $\sigma_0(p)$ in \cref{as:1.1}).

Now let $i\geq 1$, and assume that $\mathcal F(\hist{i-1}{q})$ has been defined for every $q\in V$. We will show how, given $\hist{i}{p}=\nu(h)$, the function $\mathcal F$ can compute $\sigma_i(p)$.

By definition of view, it immediately follows that the view of any node $h'\in H_h$ is contained in the view of $h$. That is, all nodes and all black and red edges of $\nu(h')$ are contained in $\nu(h)$, and the two views also agree on the labels. Thus, given the history $\hist{i}{p}=\nu(h)$, the function $\mathcal F$ can infer the view of any node $h'\in H_h$ by taking all the ascending paths in $\nu(h)$ with lower endpoint $h'$. In particular, if $\{h,h'\}$ is a black or red edge in $\nu(h)$, then $h'\in L_{i-1}$, and therefore $\mathcal F$ can determine the state $\sigma_{i-1}(q)$ of any process $q\in \rho_{i-1}^{-1}(h')$, by the inductive hypothesis.

Observe that the only black edge $\{h,h''\}\in B_h$ incident to $h$ connects it with its parent $h''\in L_{i-1}$, and $p\in \rho_{i-1}^{-1}(h'')$, due to \cref{obs:hist}. Thus, $\mathcal F$ can identify the parent of $h$ and determine $\sigma_{i-1}(p)$. Similarly, the observation multiset $o_i(p)$ can be inferred by taking the multiplicities of all red edges of the form $\{h, h'''\}\in R_h$ (cf.~the construction of the red edges in \cref{as:2.2}). Again, for each such $h'''\in L_{i-1}$, it is possible to determine the state $\sigma_{i-1}(q)$ of any process $q\in \rho_{i-1}^{-1}(h''')$. The multiset of these states (with the multiplicities inherited from $o_i(p)$) is precisely $M_i(p)$, i.e., the multiset of states that $p$ receives at round $r_i$ (cf.~the definition of $M_i(p)$ in \cref{as:1.1}).

We conclude that $\mathcal F$ can compute $\sigma_i(p)$ by first determining $\sigma_{i-1}(p)$ and $M_i(p)$, and then computing $\mathcal A(\sigma_{i-1}(p),M_i(p))$ (cf.~the definition of $\sigma_i(p)$ in \cref{as:1.1}).
\end{proof}

\begin{corollary}\label{cor:same}
During a computation in an anonymous dynamic network, at the end of round $r_i$, with $i\geq 0$, all processes represented by the same node (in the $i$th level) of the history tree have the same state.
\end{corollary}
\begin{proof}
The history at round $r_i$ of all processes represented by the node $h\in L_i$ is $\nu(h)$. Therefore, all such processes must have the same state $\mathcal F(\nu(h))$, by \cref{th:view}.
\end{proof}

\subsection{Constructing and Updating History Trees}\label{as:2.4}

We will now describe the algorithm $\mathcal A^\ast$ mentioned at the end of \cref{xs:3}. This algorithm takes as input a process $p$'s history at the end of the previous round $\xi_{i-1}(p)$, as well as the multiset of histories of neighboring processes $M_i(p)$, and constructs the new history $\xi_{i}(p)$ by \emph{merging} $\xi_{i-1}(p)$ with all the histories in $M_i(p)$.

Let us give a preliminary definition. A \emph{homomorphism} from a view $\nu(a)=(H_a, B_a, R_a, \ell_a)\in\View{\mathcal I}$ to a view $\nu(b)=(H_b, B_b, R_b, \ell_b)\in \View{\mathcal I}$ is a function $\varphi\colon H_a\to H_b$ that ``preserves structure''. That is, for all $\{h,h'\}\in B_a$, we have $\{\varphi(h),\varphi(h')\}\in B_b$; for all $h,h'\in H_a$, we have $R_a(\{h,h'\})=R_b(\{\varphi(h),\varphi(h')\})$; for all $h\in H_a$, we have $\ell_a(h)=\ell_b(\varphi(h))$.

The new history $\xi_{i}(p)$ can be constructed from $\xi_{i-1}(p)$ and $M_i(p)$ as follows. The first step is to identify the viewpoint of $\xi_{i-1}(p)$, which is the (unique) deepest node $h$ in the view; note that $h=\rho_i(p)$. The second step is to ``extend'' $\xi_{i-1}(p)$ by adding a new node $h'$, with the same label as $h$, and the new black edge $\{h,h'\}$. Let $Z_0\in \View{\mathcal I}$ be the resulting view; note that $h'$ is the (unique) child of $h$ in $Z_0$. Eventually, the node $h'$ will be the viewpoint of $\xi_{i}(p)$.

Let $\xi_{i-1}(q_1)$, $\xi_{i-1}(q_2)$, \dots, $\xi_{i-1}(q_k)$ be the histories with positive multiplicity in $M_i(p)$ (note that they must all be histories of processes at round~$r_{i-1}$: such are the messages received by $p$ at round~$r_i$). The next phase of the algorithm is to construct a sequence of views $Z_0, Z_1, Z_2, \dots, Z_k\in \View{\mathcal I}$ such that $Z_j$ is the smallest view in $\View{\mathcal I}$ that contains both $Z_{j-1}$ and $\xi_{i-1}(q_j)$, for all $1\leq j\leq k$.

In practice, the algorithm constructs $Z_j$ by starting from $Z_{j-1}$ and gradually adding the ``missing nodes'' from $\xi_{i-1}(q_j)$, at the same time constructing an (injective) homomorphism $\varphi_j$ from $\xi_{i-1}(q_j)$ to $Z_j$. First, the root of $\xi_{i-1}(q_j)$ is mapped by $\varphi_j$ to the root of $Z_{j-1}$. Then, the algorithm scans all the nodes of $\xi_{i-1}(q_j)$ level by level (i.e., doing a breadth-first traversal along the black edges). Let $v$ be a node of $\xi_{i-1}(q_j)$ encountered during the traversal, and let $v'$ be its parent in $\xi_{i-1}(q_j)$. The algorithm attempts to match $v$ with a child $w$ of the node $w'=\varphi_j(v')$ in $Z_{j-1}$. The label of $w$ should be the same as the label of $v$ and, for every red edge $\{v,v''\}$ in $\xi_{i-1}(q_j)$, connecting $v$ with a previous-level node $v''$, the edge $\{w,\varphi_j(v'')\}$ should also appear in $Z_{j-1}$ with the same multiplicity. If a node $w$ with these properties does not exist, the algorithm creates one, and then sets $\varphi_j(v)=w$. When the traversal is over, the resulting structure is $Z_j$, by definition.

Note that the final structure $Z_k$ coincides with $\xi_{i}(p)$, except for some missing red edges: these are the red edges incident to the viewpoint $h'$ which represent the messages received by $p$ at round~$r_i$. Thus, for every $1\leq j\leq k$, the viewpoint $h_j$ of $\xi_{i-1}(q_j)$ is found (as the unique deepest node in $\xi_{i-1}(q_j)$) and the red edge $\{h',\varphi_j(h_j)\}$ is added to $Z_k$, with the same multiplicity as $\xi_{i-1}(q_j)$ in $M_i(p)$. The resulting history is $\xi_{i}(p)$.

\section{Additional Lower Bounds and Counterexamples}\label{as:3}

\subsection{Naive Termination May Fail}\label{as:3.1}

The example in \cref{xfig:counter1} shows that the techniques in \cref{xs:4.1} are insufficient to formulate a correct termination condition. The first level in the leader's view where all nodes are non-branching is $L_0$, but the stabilizing algorithm in \cref{xs:4.1} computes the wrong anonymities, reporting that $n=5$ (note that $L_0$ is actually branching in the full history tree). By increasing $k$ indefinitely, we let the leader see only four nodes per level for an arbitrarily long time, while the real size of the network is kept hidden. This makes any naive termination strategy ineffective.

\begin{figure}
\centering
\includegraphics[scale=0.575]{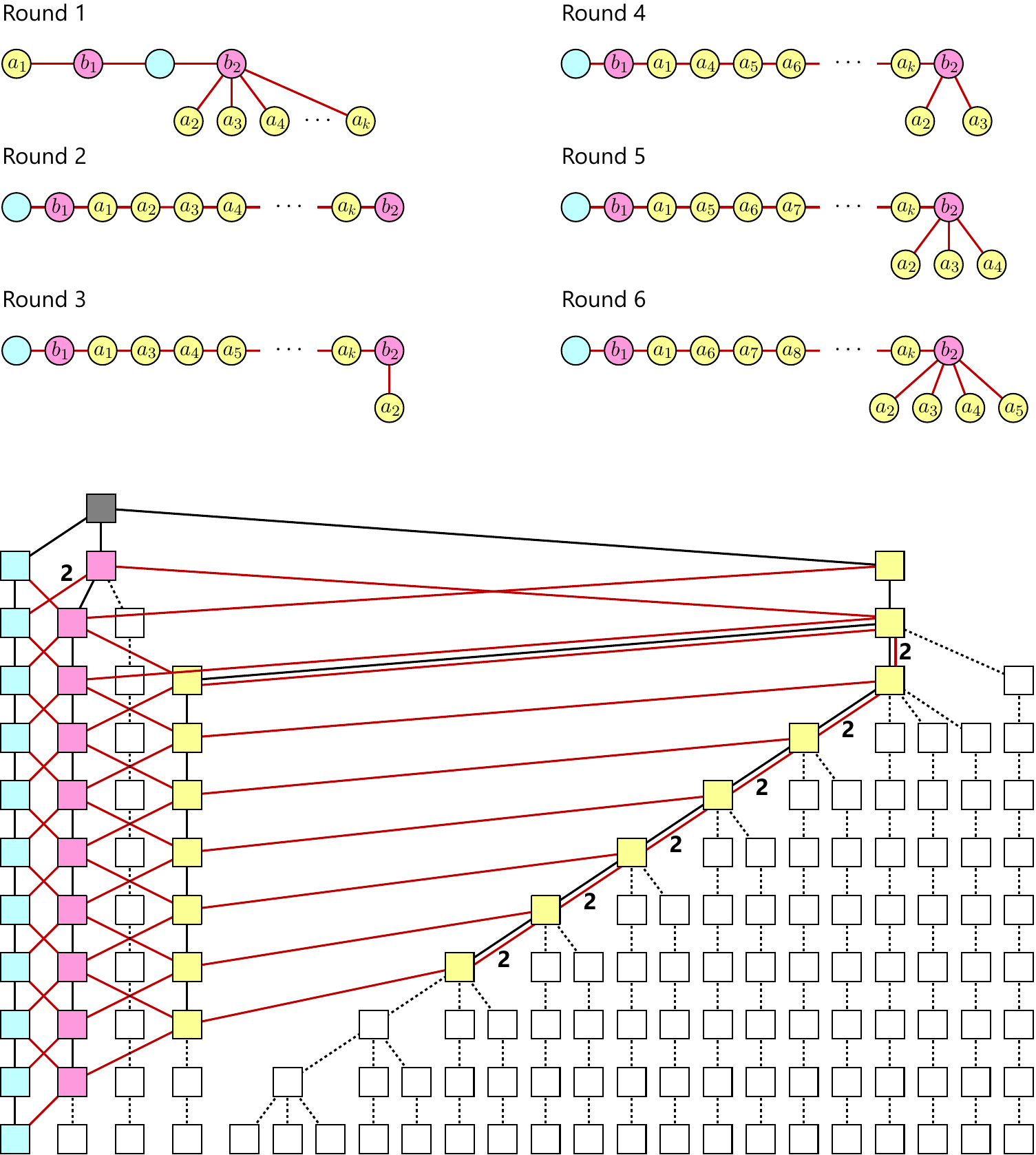}
\caption{An example of a dynamic network where the naive techniques of \cref{xs:4.1} fail to provide a termination condition. The white nodes in the history tree are not in the history of the leader at the last round; the red edges not in the view are not drawn. Same-colored processes have equal inputs. Note that, after level $L_1$, all levels in the leader's view are identical for an arbitrarily long sequence of rounds (depending on the parameter $k$).}
\label{xfig:counter1}
\end{figure}

\subsection{Worst-Case Example for the Terminating Algorithm}\label{as:3.2}

The example in \cref{xfig:counter2}, which can be easily generalized to networks of any size $n$, shows that the counting algorithm of \cref{xs:4.2} may terminate in $3n-3$ rounds. This almost matches our analysis in \cref{xth:term}, which gives an upper bound of $3n-2$ rounds.

\begin{figure}
\centering
\includegraphics[scale=0.575]{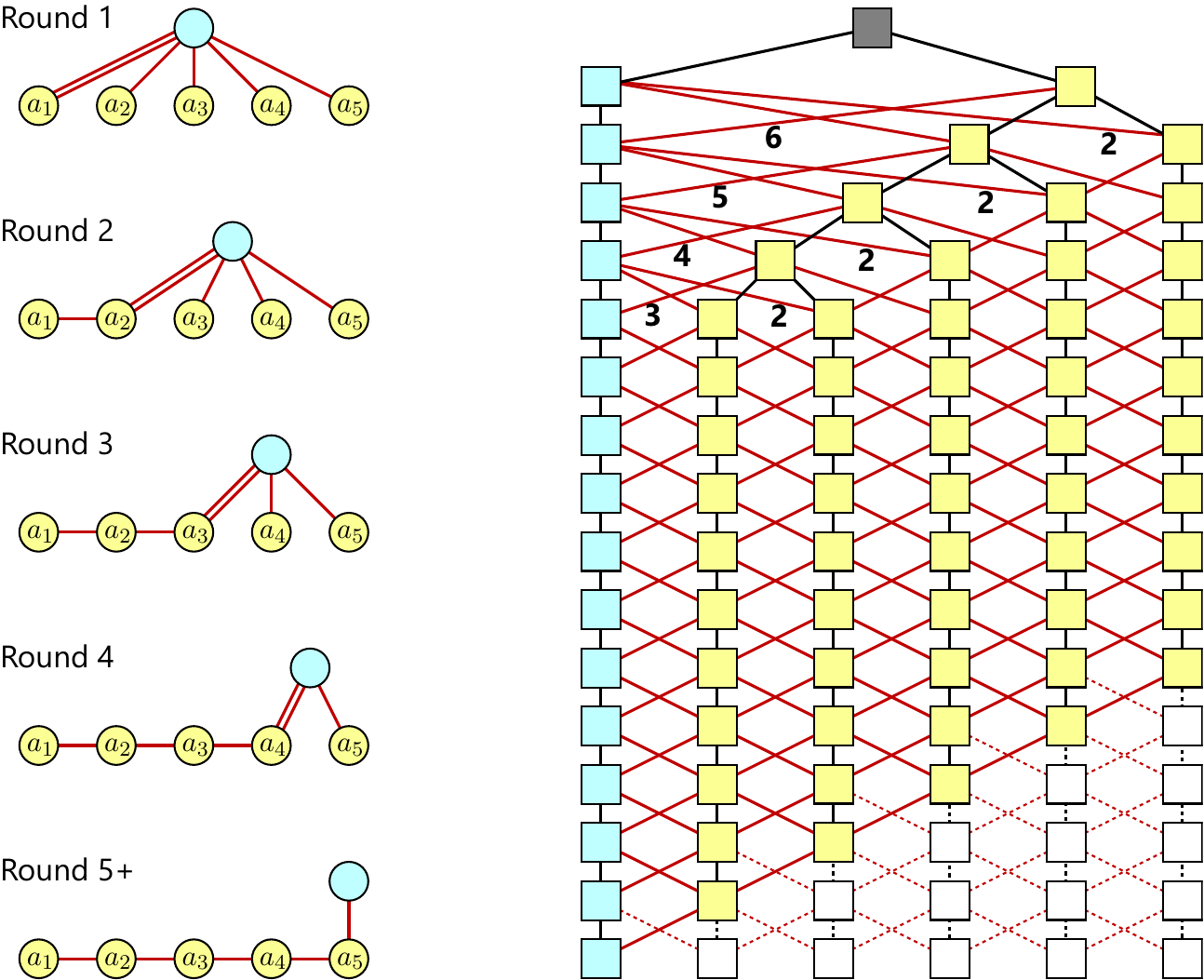}
\caption{An example of a dynamic network where the algorithm of \cref{xs:4.2} terminates in $3n-3$ rounds, and its history tree. The white nodes are not in the history of the leader at the last round.}
\label{xfig:counter2}
\end{figure}

\subsection{More General but Wearker Lower Bounds}\label{as:3.3}

We will show that a large class of multi-aggregation problems cannot be solved in an anonymous dynamic network in less than $1.5n-2$ rounds, even if the network's topology may change only once. Note that the bound of \cref{xth:lower} is better, but it only applies to the Counting Problem, and requires the network's topology to change $\Omega(n)$ times.

Our new lower bound is based on the construction of some 1-interval-connected dynamic networks whose topology changes from a cycle graph to a path graph (refer to \cref{xfig:lower3,xfig:lower4}). Given two parameters $n,m\in\mathbb N^+$ with $m<n$, we define the system $V_n=\{1,2,\dots, n\}$ and the dynamic network $\mathcal G_{n,m}=(G_i=(V_n,E_i))_{i\geq 1}$ as follows. For all $1\leq i< m$, the edge multiset $E_i$ contains all edges of the form $(j,j+1)$ with $1\leq j<n$, as well as $(1,n)$, all with multiplicity~$1$ (thus, $G_i$ is a cycle graph spanning all processes). For $i\geq m$, the topology changes slightly: the edges $(m,m+1)$ and $(1,n)$ are removed, and the edge $(m,n)$ is introduced, with multiplicity~$1$ (thus, $G_i$ is a path graph spanning all processes, with endpoints~$1$ and~$m+1$).

We will first apply our technique to the Counting Problem. Although this is redundant in light of \cref{xth:lower}, it allows us to showcase a proof pattern that applies to several more problems.
\begin{theorem}\label{th:lower}
No deterministic algorithm can solve the Counting Problem $\mathcal C_\mathcal I$ in an anonymous dynamic network of $n$ processes in less than $1.5n-2$ rounds, even if the network is 1-interval-connected, and even if there is a unique leader in the system.
\end{theorem}
\begin{proof}
Let $\mathcal I=\{L,N\}\times \mathcal I'$ be the input set, and let $x\in \mathcal I'$. Fix $m\geq 1$, and let $n=2m$ and $n'=2m+1$. Consider the two systems $V=V_n$ and $V'=V_{n'}$ with their respective dynamic networks $\mathcal G=\mathcal G_{n,m}=(G_i=(V_n,E_i))_{i\geq 1}$ and $\mathcal G'=\mathcal G_{n',m}=(G'_i=(V_{n'},E'_i))_{i\geq 1}$ (where $\mathcal G_{n,m}$ and $\mathcal G_{n',m}$ are as defined above). We further define the input assignments $\lambda\colon V\to \mathcal I$ and $\lambda'\colon V'\to \mathcal I$, which assign the input $(L,x)$ to process~$1$ and the input $(N,x)$ to all other processes. Thus, in both systems, process~$1$ is the leader, and all other processes are anonymous. We denote as $\mathcal H$ (respectively, $\mathcal H'$) the history tree associated with $\lambda$ and $\mathcal G$ (respectively, $\lambda'$ and $\mathcal G'$), while $\xi_i(p)$ and $\xi'_i(p')$ denote the histories of processes $p\in V$ and $p'\in V'$, respectively, at round $r_i$.

We define the \emph{$d$-neighborhood} $\mathcal N_i(p,d)$ of a process $p\in V$ at round $r_i$, with $i\geq 1$, as the subgraph of $G_i$ induced by the processes that have distance at most $d$ from $p$ in $G_i$ (similarly, we define the $d$-neighborhood $\mathcal N'_i(p',d)$ of a process in $p'\in V'$). We say that $\mathcal N_i(p,d)$ is \emph{equivalent} to $\mathcal N'_i(p',d)$ if there is a graph isomorphism $\varphi$ from $\mathcal N_i(p,d)$ to $\mathcal N'_i(p',d)$ that preserves inputs; that is, $\lambda(q)=\lambda'(\varphi(q))$ for all processes $q$ in $\mathcal N_i(p,d)$.

Recall that, for the first $m-1$ rounds, the topology of both networks is a cycle graph. Thus, all processes at the same distance from their respective leader have equivalent $(m-1)$-neighborhoods throughout the first $m-1$ rounds, and therefore also have equal histories. That is, there is a function $\delta\colon V\to V'$ such that $\xi_i(p)=\xi'_i(\delta(p))$ for all $p\in V$ and $1\leq i<m$. Namely, $\delta(p)=p$ for all $1\leq p\leq m$, and $\delta(n-p)=n'-p$ for all $0\leq p<m$.

Starting at round~$r_m$, the topology of both networks changes to a path graph with the leader at one endpoint. These path graphs have the property that the process at distance $d$ from the leader in $V$ has the same history as the process at distance $d$ from the leader in $V'$ for all $0\leq d<n$. Thus, the two leaders will keep having equal histories for the following $n-1=2m-1$ rounds. We conclude that the leaders of the two systems have equal histories up to round $r_{3m-2}$.

Assume for a contradiction that there are computation parameters (as defined in \cref{as:1.1}) that cause all processes in both systems to output the inventory of their respective input assignments, $\mu_\lambda$ and $\mu_{\lambda'}$, thus solving $\mathcal C_\mathcal I$, in less than $1.5n-1=\lfloor 1.5n'\rfloor-2 = 3m-1$ rounds. Since the leaders of $V$ and $V'$ have equal histories throughout the first $3m-2$ rounds, by \cref{th:view} they must have equal states, as well. Thus, both leaders must give the same output upon termination, implying that $\mu_\lambda=\mu_{\lambda'}$. This is a contradiction, because $\mu_\lambda((N,x))=n-1\neq n=\mu_{\lambda'}((N,x))$.
\end{proof}

We can extend \cref{th:lower} to several other (multi-aggregation) problems. For example, fix $n,k\in \mathbb N^+$, let $m=\lfloor n/2\rfloor$, and consider the two dynamic networks $\mathcal G=\mathcal G_{n,m}$ and $\mathcal G'=\mathcal G_{2n+k-1,m}$ for the systems $V=V_n$ and $V'=V_{2n+k-1}$, respectively. Let us assign a leader input to process~$1$ in both systems, any arbitrary $k$-tuple $((N,x_1),(N,x_2),\dots, (N,x_k))$ of inputs to processes $n+1$, $n+2$, \dots, $n+k$ in $V'$, and another input $(N,x)$ (not in the $k$-tuple) to all other processes.

Now, the same argument used for proving \cref{th:lower} shows that the two leaders have equal histories for the first $3m-2$ rounds. Thus, the leaders are not only unable to count the number of processes in less than $1.5n-2$ rounds, but are also unable to tell whether the system contains any process at all with inputs in the chosen $k$-tuple.

\begin{figure}
\centering
\includegraphics[scale=0.575]{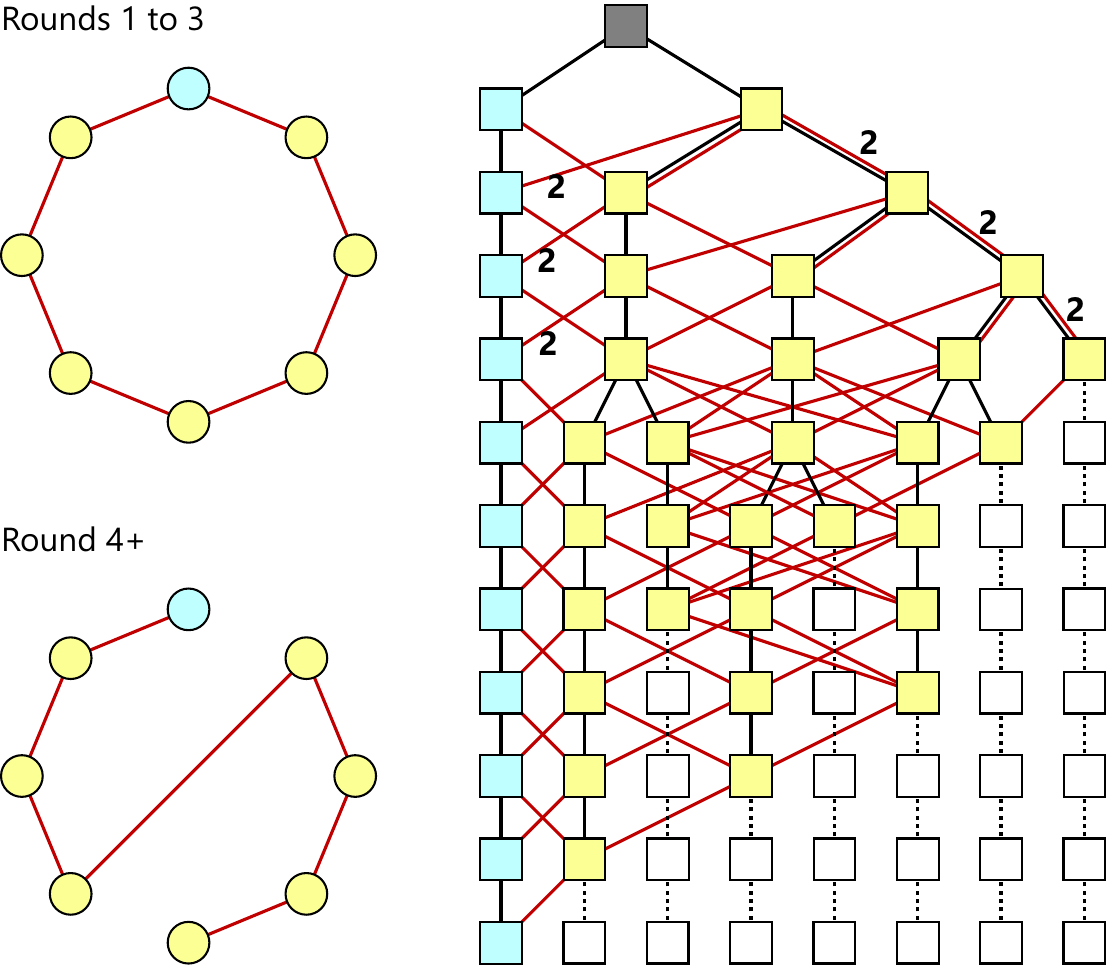}
\caption{The dynamic network $\mathcal G_{n,m}$ with $m=4$ and $n=8$, and its history tree. The red edges not in the view of the leader have not been drawn.}
\label{xfig:lower3}
\end{figure}

\begin{figure}
\centering
\includegraphics[scale=0.575]{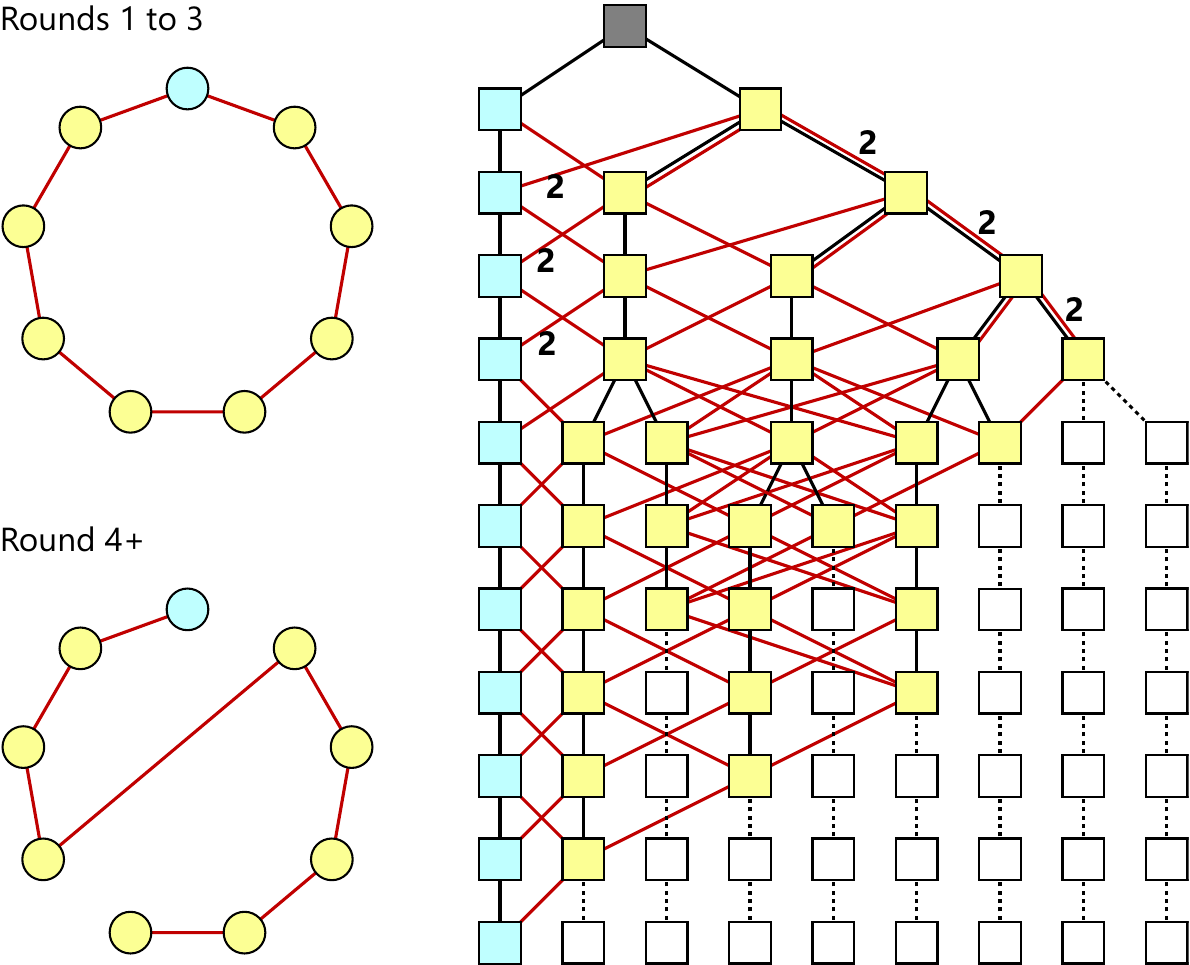}
\caption{The dynamic network $\mathcal G_{n,m}$ with $m=4$ and $n=9$, and its history tree. Observe that the leaders of this network and of the network in \cref{xfig:lower3} have isomorphic histories at round $3m-2=10$ (cf.~\cref{th:lower}).}
\label{xfig:lower4}
\end{figure}

\section{Survey of Related Work}\label{as:4}

We examine related work on counting and related problems by first discussing the case of dynamic networks with unique IDs, then the case of static anonymous networks, and finally the case of interval-connected anonymous networks. We conclude the section by examining the average consensus problem, which is deeply related to counting.

\subsection{Dynamic Networks with IDs}

The problem of counting the size of a dynamic network has been first studied by the peer-to-peer systems community  \cite{LKM06}. In this case having an exact count of the network at a given time is impossible, as processes may join or leave in an unrestricted way. Therefore, their algorithms mainly focus on providing estimates on the network size with some guarantees. 
The most related is the work that introduced 1-interval-connected networks \cite{KLO10}. They show a counting algorithm that terminates in at most $n+1$ rounds when messages are unrestricted and in $O(n^2)$ rounds when the message size is $O(\log n)$ bits.  The techniques used heavily rely on the presence of unique IDs and cannot be extended to our settings. 

\subsection{Anonymous Static Networks}
The study of computability on anonymous networks has been pioneered by Angluin in \cite{A80} and it has been a fruitful research topics for the last 30 years \cite{A80,BV01,CDS06,CGM08,JMM12,FPP00,SUW15,YK88}. A key concept in anonymous networks is the symmetry of the system; informally, it is the indistinguishability of nodes that have the same view of the network. As an example, in an anonymous static ring topology, all processes will have the exact same view of the system, and such a view does not change between rings of different size. Therefore, non-trivial computations including counting are impossible on rings, and some symmetry-breaking assumption is needed (such as a leader \cite{FPP00}). The situation changes if we consider topologies that are asymmetric. As an example, on a wheel graph the central node has a view that is unique, and this allows for the election of a leader and the possibility, among other tasks, of counting the size of the network. 

Several tools have been developed to characterize what can be computed on a given network topology (examples are views \cite{YK88} or fibrations \cite{BV02}). Unfortunately, these techniques are usable only in the static case and are not defined for highly dynamic systems like the ones studied in our work. Regarding the counting problem in anonymous static networks with a leader, \cite{MCS13} gives a counting algorithm that terminates in at most $2n$ rounds.

\subsection{Counting in Anonymous Interval-Connected Networks}
The papers that studied counting in anonymous dynamic networks can be divided into two periods. A first series of works~\cite{CMM18,DBBC14a,DBBC14b,MCS13} gave solutions for the counting problem assuming some initial knowledge on the possible degree of a processes. As a matter of fact \cite{MCS13} conjectured that some kind of knowledge was necessary to have a terminating counting algorithm. A second series of works~\cite{DB15,KM18,KM20,KM21,KM22} has first shown that counting was possible without such knowledge, and then has proposed increasingly faster solutions. We remark that all these papers assume that a leader (or multiple leaders in \cite{KM21}) is present. This assumption is needed to break the symmetry. 

\paragraph{Counting with knowledge on the degrees.}
Counting in interval-connected anonymous networks was first studied in \cite{MCS13}, where it is observed that a leader is necessary to solve counting in static (and therefore also dynamic) anonymous networks (this result can be derived from previous works on static networks such as \cite{BV02,YK88}). The paper does not give a counting algorithm but it gives an algorithm that is able to compute an upper bound on the network size. Specifically, \cite{MCS13} proposes an algorithm that, using an upper bound $d$ on the maximum degree that each process will ever have in the network, calculates an upper bound $U$ on the size of the network; this upper bound may be exponential in the actual network size ($U \leq d^{n}$). 

Assuming the knowledge of an upper bound on the degree, \cite{DBBC14a} given a counting algorithm that computes $n$. Such an algorithm is really costly in terms of rounds; it has been shown in \cite{CMM18} to be doubly exponential in the network size.  The algorithm proposes a mass distribution approach akin to local averaging \cite{T84}.

An experimental evaluation of the algorithm in \cite{DBBC14a} can be found in \cite{DBCB13}. The result of  \cite{DBBC14a} has been improved in \cite{CMM18}, where, again assuming knowledge of an upper bound $d$ on the maximum degree of a node, an algorithm is given that terminates in $ O\left(n(2d)^{n + 1} \frac {\log n} { \log {d}}\right)$ rounds.
A later paper~\cite{DBBC14b} has shown that counting is possible when each process knows its degree before starting the round (for example, by means of an oracle). In this case, no prior global upper bound on the degree of processes is needed. \cite{DBBC14b} only show that the algorithm eventually terminates but does not bound the termination time. 

We remark that all the above works assume some knowledge on the dynamic network, as an upper bound on the possible degrees, or as a local oracle.  Moreover, all of these works give exponential-time algorithms. 

\paragraph{Counting without knowledge on the degrees.}
The first work proposing an algorithm that does not require any knowledge of the network was  \cite{DB16}. The paper proposed an exponential-time algorithm that terminates in $O\left(n^{n+4}\right)$ rounds. Moreover, it also gives an asymptotically  optimal algorithm for a particular category of networks (called persistent-distance). In this type of network, a node never changes its distance from the leader. 

This result was improved in \cite{KM18,KM20}, which presented a polynomial-time counting algorithm.  The paper proposes Methodical Counting, an algorithm that counts in $O (n^5\log^2(n))$ rounds. Similar to   \cite{DBBC14a,DBBC14b}, the paper uses a mass-distribution process that is coupled with a refined analysis of convergence time and clever techniques to detect termination. The paper also notes that, using the same algorithm, all algebraic and boolean functions that depend on the initial number of processes in a given state can be computed. 
The authors of  \cite{KM18,KM20} extended their result to work in networks where $l \geq 1$ leaders are present (with $l$ known in advance) in \cite{KM22}. They  create an algorithm that terminates in $O \left(\frac {n^{4+ \epsilon} } {l} \log^{3} (n)\right)$ rounds, for any $\epsilon >0$. In particular, when $ l = 1 $, this results improves the running time of \cite{KM18,KM20}.

Finally,  in \cite{KM21}, they show a counting algorithm parameterized by the isoperimetric number of the dynamic network. The technique used is similar to \cite{KM18,KM20}, and it uses the knowledge of the isoperimetric number to shorten the termination time. Specifically,  for adversarial graphs (i.e., with non-random topology) with $l$ leaders ($l$ is assumed to be known in advance), they give an algorithm terminating in $O \left(\frac {n ^ {3+ \epsilon}} {l {i_ {min} }^2} \log^{3}(n)\right)$ rounds, where $i_{min}$ is a known lower bound on the isoperimetric number of the network. This improves the work in \cite{KM22}, but only in graphs where $i_ {min} $ is $ \omega ({1}/ {\sqrt {n}})$. The authors also study various types of graphs with stochastic dynamism; we remark that in this case they always obtain superlinear results, as well. The best case is that of Erdős--Rényi, graphs where their algorithm terminates in $ O\left (\frac {n ^ {1+ \epsilon}} {l {p_ {min}} ^ 2} \log ^ {5} (n) \right) $ rounds; here $ p_{min}$ is the smallest among the probabilities of creating an arc on all rounds. Specifically, if $ p_{min} =O ({1}/{n}) $, their algorithm is at least cubic.
 
Summarizing, to the best of our knowledge, no linear-time algorithm is  known for anonymous dynamic networks, even when additional knowledge is provided (e.g., the isoperimetric number~\cite{KM21}), or when the topology of the network is random and not worst-case. Notice that a random topology is a really powerful assumption in anonymous networks, as it is likely to greatly help in breaking symmetry.

We stress that our algorithms run in linear time without requiring any prior knowledge, and works for worst-case dynamic topologies (and thus, also in the case of a random dynamicity).  

\paragraph{Lower bounds on counting.}

From \cite{KLO10}, a trivial lower bound of $n-1$ rounds can be derived, as counting obviously requires information from each node to be spread in the network. Interestingly, prior to our work, little was known apart from this result. The only other lower bound is in~\cite{DB15}, which shows a specific category of anonymous dynamic networks with constant temporal diameter (the time needed to spread information from a node to all others is at most 3 rounds), but where counting requires $\Omega(\log n)$ rounds. 
Therefore, our lower bound of roughly $2n$ rounds greatly improves on the previously known lower bounds for the general case (networks that have an arbitrary temporal diameter). 

\subsection{Average Consensus}
Counting is deeply related to the average consensus problem. In this problem, each process $v_i$ starts with an input value $ x_i (0) $, and processes must calculate the average of these initial values. Note that the Generalized Counting trivially leads to a solution for the average consensus. 

The first paper to propose a decentralized solution was \cite{T84}. It introduced an approach called ``local averaging", where each process updates its local value $ x_i (r) $ at every round $r$ as follows:
\begin{center}
$\displaystyle x_i (r) = \sum _ {\forall {v_j \in N (r, v_i) \cup \{v_i \}}} a_ {ij} (r) \cdot x_j ( r-1).$
\end{center}
The value $ a_{ij} (r) $ is taken from a weight matrix that can model a dynamic graph.  The $ \epsilon $-convergence of this algorithm is defined as the time it takes to be sure that the maximum discrepancy between the local value of a node and the mean is at most $\epsilon$ times the initial discrepancy. That is, if the mean is $m=\sum_ {i \in V} x_i (0)/ {| V |}$, the following should hold:

$$  \frac{ \text{max}_{i} \left\{|x_i (r) - m| \right\}}{\text{max}_{i} \left\{|x_i (0) - m| \right\}}\leq \epsilon.$$

The local averaging approach has been studied in depth, and several upper and lower bounds for $\epsilon$-convergence are known for both static and dynamic networks \cite{OT09,OT11}.
The procedure $ \epsilon $-converges in $ O (|V| ^ {3} \log (1/{\epsilon})) $ rounds if, at every round, the weight matrix $ A (r)$ such that $ (A (r))_{ij} = a_{ij} (r)$ is doubly stochastic (i.e., the sum of the values on rows and columns is $1$) \cite{NOOT09}. In a dynamic network, it is possible to have doubly stochastic weight matrices when an upper bound on the node degree is known \cite{C11}. 
We remark that local averaging algorithms do not need unique IDs, and thus work in anonymous dynamic networks; we also note that the process does not explicitly terminate, but just converges to the average of the inputs. 

To the best of our knowledge, there is no solution to average consensus based on local averaging or other techniques with a linear convergence time. Our algorithm shows that solving the average consensus in 1-interval-connected anonymous networks with a leader is possible in linear time with strict termination and perfect convergence ($\epsilon=0$). Moreover, our lower bound on counting is also a lower bound on average consensus (if the leader starts with a value of $1$ and all other processes with $0$'s, the average converges to the inverse of the size). Thus, our paper shows that a lower bound on convergence of any average consensus algorithm is roughly $2n$ in anonymous dynamic networks. 

\clearpage
\addcontentsline{toc}{section}{References}
\bibliographystyle{plainurl}
\bibliography{anonymity}
\end{document}